\documentclass[11pt]{article}

\usepackage[utf8]{inputenc} 
\usepackage[T1]{fontenc}    
\usepackage{hyperref}       
\usepackage{url}            
\usepackage{booktabs}       
\usepackage{amsfonts}       
\usepackage{nicefrac}       
\usepackage{microtype}      
\usepackage[dvipsnames]{xcolor}         

\usepackage{complexity}

\let\poly\relax
\newcommand{\poly}{\mathrm{poly}}

\usepackage{fullpage}

\usepackage{colortbl}

\usepackage{authblk}

\usepackage{mathtools}
\usepackage{amsmath}
\usepackage{amssymb}
\usepackage{amsthm}

\usepackage{bbm}

\newcommand{\hvx}{\widehat{\vx}}
\newcommand{\hvy}{\widehat{\vy}}
\newcommand{\hvmu}{\widehat{\vmu}}

\usepackage{mleftright}
\usepackage{thm-restate}
\usepackage{bm}

\definecolor{lightgray}{gray}{0.85}
\usepackage{tikz}
\usepackage{colortbl}

\newcommand{\declarecolor}[2]{\definecolor{#1}{RGB}{#2}\expandafter\newcommand\csname #1\endcsname[1]{\textcolor{#1}{##1}}}
\declarecolor{White}{255, 255, 255}
\declarecolor{Black}{0, 0, 0}
\declarecolor{Maroon}{128, 0, 0}
\declarecolor{Coral}{255, 127, 80}
\declarecolor{Red}{182, 21, 21}
\declarecolor{LimeGreen}{50, 205, 50}
\declarecolor{DarkGreen}{0, 80, 0}
\declarecolor{Purple}{146, 42, 158}
\declarecolor{Navy}{0, 0, 128}
\declarecolor{LightBlue}{84, 101, 202}
\definecolor{mydarkblue}{rgb}{0,0.08,0.45}
\hypersetup{ %
    pdftitle={},
    pdfkeywords={},
    pdfborder=0 0 0,
    pdfpagemode=UseNone,
    colorlinks=true,
    linkcolor=Navy,
    citecolor=DarkGreen,
    filecolor=Purple,
    urlcolor=Black,
}

\usepackage[customcolors]{hf-tikz}

\tikzset{style green/.style={
    set fill color=green!100!lime!40,
    set border color=white,
  },
  style cyan/.style={
    set fill color=cyan!90!blue!40,
    set border color=white,
  },
  style orange/.style={
    set fill color=orange!80!red!60,
    set border color=white,
  },
  style red/.style={
    set fill color = red!50,
    set border color=white,
  },
  hor/.style={
    above left offset={-0.15,0.31},
    below right offset={0.15,-0.125},
    #1
  },
  ver/.style={
    above left offset={-0.1,0.3},
    below right offset={0.15,-0.15},
    #1
  }
}

\usepackage[ruled,linesnumbered]{algorithm2e}

\usepackage[font=small,labelfont=bf]{caption}

\SetKwComment{Comment}{/* }{ */}

\usepackage[capitalize,noabbrev,nameinlink]{cleveref}

\usepackage{MnSymbol}

\usepackage[natbib,maxcitenames=6,maxnames=6,style=alphabetic]{biblatex}

\bibliography{main}

\let\E\relax

\usepackage{enumitem}

\usepackage{blkarray}

\newcommand{\sumx}{\vx_U}
\newcommand{\sumy}{\vy_U}

\newcommand{\setT}{\mathcal{T}_U}

\newcommand{\accept}{\textsf{Accept}}
\newcommand{\reject}{\textsf{Reject}}

\newcommand{\vxdef}{\mathring{\vec{x}}}
\newcommand{\vydef}{\mathring{\vec{y}}}

\DeclareMathOperator{\pr}{\mathbb{P}}

\newcommand{\reg}{\mathsf{Reg}}

\newcommand{\vu}{\vec{u}}

\newcommand*{\N}{{\mathbb{N}}}

\let\R\relax
\newcommand*{\R}{{\mathbb{R}}}
\newcommand*{\E}{{\mathbb{E}}}

\newcommand*{\cG}{{\mathcal{G}}}

\newcommand{\defeq}{\coloneqq}

\newcommand{\OV}{\textsf{OV}}

\newcommand{\vx}{\vec{x}}
\newcommand{\vy}{\vec{y}}
\newcommand{\vz}{\vec{z}}

\newcommand{\vmu}{\vec{\mu}}

\DeclareMathOperator{\sw}{\mathsf{SW}}

\newcommand{\tilmu}{\tilde{\mu}}

\newcommand{\tilE}{\tilde{\E}}

\newcommand{\enumhard}{\textsf{EnumHard}}
\newcommand{\soshard}{\textsf{SoSHard}}

\renewcommand{\vec}[1]{\bm{#1}}
\newcommand{\mat}[1]{\mathbf{#1}}

\theoremstyle{plain}
\newtheorem{theorem}{Theorem}[section]
\newtheorem{lemma}[theorem]{Lemma}

\newtheorem{proposition}[theorem]{Proposition}
\newtheorem{fact}[theorem]{Fact}

\newtheorem{claim}[theorem]{Claim}

\theoremstyle{definition}
\newtheorem{definition}[theorem]{Definition}

\theoremstyle{remark}

\title{Computational Lower Bounds for Regret Minimization\\ in Normal-Form Games}

\author[1]{Ioannis Anagnostides}
\author[2]{Alkis Kalavasis}
\author[3]{Tuomas Sandholm}

\affil[1,3]{Carnegie Mellon University}
\affil[2]{Yale University}
\affil[3]{Strategy Robot, Inc.}
\affil[3]{Strategic Machine, Inc.}
\affil[3]{Optimized Markets, Inc.}
\affil[ ]{\texttt{\{ianagnos,sandholm\}}\texttt{@cs.cmu.edu}, \texttt{alkis.kalavasis}\texttt{@yale.edu}}

\date{}

\begin{document}

\maketitle

\begin{abstract}
    A celebrated connection in the interface of online learning and game theory establishes that players minimizing \emph{swap regret} converge to \emph{correlated equilibria (CE)}---a seminal game-theoretic solution concept. Despite the long history of this problem and the renewed interest it has received in recent years, a basic question remains open: how many iterations are needed to approximate an equilibrium under the usual \emph{normal-form} representation? In this paper, we provide evidence that existing learning algorithms, such as multiplicative weights update, are close to optimal. In particular, we prove lower bounds for the problem of computing a CE that can be expressed as a uniform mixture of $T$ product distributions---namely, a uniform \emph{$T$-sparse CE}; such lower bounds immediately circumscribe (computationally bounded) regret minimization algorithms in games. Our results are obtained in the algorithmic framework put forward by Kothari and Mehta (STOC 2018) in the context of computing Nash equilibria, which consists of the \emph{sum-of-squares (SoS)} relaxation in conjunction with oracle access to a \emph{verification oracle}; the goal in that framework is to lower bound either the degree of the SoS relaxation or the number of queries to the verification oracle. Here, we obtain two such hardness results, precluding computing i) uniform $\log n$-sparse CE when $\epsilon =\poly(1/\log n)$ and ii) uniform $n^{1 - o(1)}$-sparse CE when $\epsilon = \poly(1/n)$.
\end{abstract}

\pagenumbering{gobble}

\clearpage

\pagenumbering{arabic}


\section{Introduction}

A celebrated line of research in the interface of algorithmic game theory and online learning revolves around the repeated interaction of multiple players in a game. Much of this theory stems from the realization that players engaging rationally---in that their behavior is consistent with some notion of hindsight rationality or \emph{no-regret}---converge to a certain game-theoretic solution concept known as \emph{(coarse) correlated equilibrium (CE)}~\citep{Hart00:Simple,Foster97:Calibrated}. From an algorithmic standpoint, perhaps the most pressing question emerging from that connection pertains to the number of iterations needed to approximate an equilibrium. One answer put forward in the online learning literature postulates that each player is facing an adversarial environment, a regime which is by now well-understood commencing from some early influential work~\citep{Littlestone94:Weighted,Littlestone87:Learning}. When learning in games, however, players instead interact with other learning algorithms; the obvious concern thus is that the predictions of the traditional no-regret framework are overly pessimistic. Indeed, it turns out that barriers ingrained in the adversarial regime can be circumvented when specialized algorithms are in place (\emph{e.g.},~\citep{Syrgkanis15:Fast,Rakhlin13:Optimization,Daskalakis15:Near,Daskalakis21:Near,Erez23:Regret,Daskalakis22:Fast,Cai24:Near}). As a notable example, \citet{Daskalakis21:Near} showed that a simple variant of the celebrated multiplicative weights update algorithm guarantees \emph{(external) regret} growing only as $O(\log^4 T \log n)$ after $T$ repetitions of any game with a constant number of players, where $n$ here and throughout represents the number of actions of each player; this is a significant improvement over the $\Omega(\sqrt{T \log n})$ information-theoretic barrier when facing an adversary~\citep{Cesa-Bianchi06:Prediction}. 

Practical experience also suggests a considerable gap between existing theoretical predictions and observed behavior for no-regret dynamics, manifested in attaining remarkable performance even in enormous games~\citep{Bowling15:Heads}. It is worth stressing that even though a correlated equilibrium can be computed exactly via linear programming~\citep{Papadimitriou08:Computing}, no-regret dynamics are in many ways more appealing from a computational standpoint, mostly attributed to their scalability and their decentralized nature; it has also been argued that they constitute more plausible models of learning~\citep{Sandholm10:Population}. Despite the rich history of this foundational problem, the number of iterations needed for no-regret learners to approximate an equilibrium remains an outstanding open problem.

In this paper, we study the complexity of no-regret learning in games from a \emph{computational} perspective, under the premise that learners employ polynomial-time algorithms to update their strategies. At first glance, it might be unclear why and how computation circumscribes no-regret learning. This can be grasped by considering first an extreme case: what prevents \emph{both} learners from incurring small regret after a \emph{single} iteration? The answer is that this would result in a \emph{Nash equilibrium (NE)} of the underlying game---a stronger notion than correlated equilibria, which can be immediately ruled out in light of well-known computational barriers in general-sum games~\citep{Rubinstein16:Settling,Chen09:Settling,Daskalakis08:Complexity,Boodaghians20:Smoothed,Deligkas22:Pure,Etessami07:Complexity,Kothari18:Sum}. But can those lower bounds be extended for multiple iterations of no-regret learning? To address this question, we investigate the complexity of computing a correlated equilibrium under the additional constraint that it can be expressed as a uniform mixture of $T$ \emph{product distributions}---henceforth, \emph{uniform $T$-sparse} CE (\Cref{def:sparse}). Intractability concerning sparse CE immediately lower bounds the number of iterations for computationally bounded no-regret learners. Further, a compelling aspect of this approach is that, unlike other lower bounds based on query complexity (discussed in~\Cref{sec:related}), it applies even if the players know the game upfront and can coordinate prior to the learning phase. 

This program has been followed with success starting from the work of~\citet{Foster23:Hardness} in the context of Markov (aka. stochastic) games, and continued by~\citet{Peng24:Complexity} targeting games with imperfect information. Yet, no progress has been made in understanding the complexity of sparse CE in \emph{normal-form games}, which constitutes the canonical representation treated in the literature. It is worth stressing here that although computing $2$-sparse CE appears to be a similar problem to computing NE (that is, $1$-sparse CE), adapting existing reductions based on the latter to the former turns out to be particularly challenging; relaxing players' independence---as in Nash equilibria---immediately introduces considerable technical obstacles.

Our contribution here is to tackle those challenges and provide strong evidence for the intractability of computing sparse CE in normal-form games. In particular, we extend the \emph{sum-of-squares (SoS)}-based lower bounds of~\citet{Kothari18:Sum} from Nash equilibria to sparse CE under a broad sparsity regime.

\paragraph{The algorithmic framework} We operate in the algorithmic framework put forward by~\citet{Kothari18:Sum} revolving around SoS, a sequence of increasingly more powerful semidefinite programs (SDPs). The SoS hierarchy has proven to be a remarkably effective technique for algorithm design in many fundamental problems from diverse areas (\emph{e.g.},~\citep{Arora15:Subexponential,Barak11:Rounding,Ma16:Polynomial,Harrow16:Tight,Barak15:Dictionary,Barak14:Rounding,Arora09:Expander,Kothari22:Polynomial,Barak14:Sum}). As a result, in light of the power of the framework, an SoS lower bound serves as strong evidence for the intractability of a problem. Such a lower bound typically manifests itself in the form of an \emph{integrality gap}, precluding approximating a certain objective function without ascending excessively high in the SoS hierarchy.

When it comes to equilibrium computation, however, there is no underlying objective; and the decision version of the problem is also of little use since, by virtue of its totality, a solution always exists. \citet{Kothari18:Sum} propose to address such issues by instead relying on what they refer to as \emph{rounding gaps}. The idea here is as follows. An SoS relaxation will output a relaxed solution with respect to our problem of interest---namely, $T$-sparse CE. Such a relaxation will generally not be a legitimate solution, so an additional step---known as \emph{rounding}---is required so as to identify an actual equilibrium. Now, for this approach to be meaningful, it is necessary to impose constraints on how the rounding algorithm operates, for otherwise it can simply ignore the relaxed solution and compute an equilibrium from scratch by accessing the game; proving lower bounds against such algorithms is precisely what we set out for in the first place. Following~\citet{Feige16:Oblivious}, \citet{Kothari18:Sum} address this issue by restricting the rounding algorithm to be \emph{oblivious}, in that the true solution can only depend on the relaxed one. This algorithmic framework, although restricted, is powerful enough to capture several famous algorithms~\citep{Feige16:Oblivious}, including threshold rounding for vertex cover~\citep{Hochbaum82:Approximation}; randomized rounding for set cover~\citep{Raghavan87:Randomized}; random hyperplane rounding for maximum cut~\citep{Goemans95:Improved}; and welfare maximization for fractionally subadditive (XOS) and submodular valuations~\citep{Feige09:Maximizing,Feige10:Submodular}. To strengthen the rounding algorithm and capture more existing techniques, \citet{Kothari18:Sum} also allow it to adaptively produce a list of candidate solutions---instead of a single one---by checking whether one of them is indeed a solution through a \emph{verification oracle}; this enables capturing enumeration techniques over a restricted search space, which have been successful in equilibrium computation problems---most notably, by~\citet{Lipton03:Playing}. A lower bound in this framework consists of proving that either the level---aka. the \emph{degree}---in the SoS hierarchy is high---thereby rendering the corresponding SDP out of reach---or the number of queries to the verification oracle is prohibitively large.

\subsection{Preliminaries}
\label{sec:init-prels}

To describe our results, we first need to formally introduce the problem of interest and the algorithmic framework outlined above; further background is provided later in~\Cref{sec:prels}. The familiar reader can mostly skim the upcoming paragraphs leading to~\Cref{sec:results} for our notation.

\paragraph{Two-player games} We consider two-player games represented in normal form. (Since we are aiming to prove lower bounds, concentrating on two-player games will only make the results stronger.) Here, each player has a finite set of available actions; without any loss of generality, we may and will assume that the set of actions of each player is $[n] \defeq \{1, 2, \dots, n\}$. Under a pair of actions $(i, j) \in [n] \times [n]$, the utility of the players is given by $\mat{R}_{i, j}$ and $\mat{C}_{i, j}$, respectively, where $\mat{R}, \mat{C} \in \R^{n \times n}$ are the payoff matrices of the game given as part of the input; by convention, we will refer to the players as Player $x$ (for the ``row player'') and Player $y$ (for the ``column player''), respectively. $\mathcal{G} \defeq (\mat{R}, \mat{C})$ will sometimes be referred to as an $n \times n$ game. Players can randomize by selecting as strategy a probability distribution, a point in $\Delta^n \defeq \{  \vx \in \R^n_{\geq 0} : \sum_{i=1}^n \vx_i = 1 \}$. The expected utility under a pair of mixed strategies $(\vx, \vy) \in \Delta^n \times \Delta^n$ is given by $\E_{(i,j) \sim (\vx, \vy)} \mat{R}_{i, j} = \langle \vx, \mat{R} \vy \rangle$ and $\E_{(i,j) \sim (\vx, \vy)} \mat{C}_{i, j} = \langle \vx, \mat{C} \vy \rangle$ for Player $x$ and Player $y$, respectively.

\paragraph{Correlated equilibria and sparsity} We next recall the notion of a \emph{correlated equilibrium (CE)}~\citep{Aumann74:Subjectivity}. Central to this definition is the set of \emph{swap deviations} $\Phi_{\text{swap}}$, which contains all functions mapping $[n]$ to $[n]$. Even though there are $n^n$ such functions, there is still an efficient algorithm minimizing $\Phi_{\text{swap}}$-regret~\citep{Blum07:From}. In fact, we will show that lower bounds persist even if one considers a certain subset of deviations $\Phi \subseteq \Phi_{\text{swap}}$ with polynomial (in $n$) size (defined later in~\Cref{section:CE-Regret}).

\begin{definition}
    \label{def:CE}
    A distribution $\vmu$ on $[n] \times [n]$ is an \emph{$\epsilon$-correlated equilibrium ($\epsilon$-CE)} if for any deviations $\phi_x \in \Phi$ and $\phi_y \in \Phi$,
    \begin{equation}
        \label{eq:CE}
        \E_{(i, j) \sim \vmu} \mat{R}_{i, j} \geq \E_{(i, j) \sim \vmu} \mat{R}_{\phi_x(i), j} - \epsilon \quad \text{and} \quad \E_{(i, j) \sim \vmu} \mat{C}_{i, j} \geq \E_{(i, j) \sim \vmu} \mat{C}_{i, \phi_y(j)} - \epsilon.
    \end{equation}
\end{definition}

\emph{Coarse} correlated equilibria (CCE) relax CE by instead imposing~\eqref{eq:CE} only for \emph{external deviations} $\Phi_{\text{ext}} \defeq \{ \phi \in \Phi_{\text{swap}} : \exists i' \in [n] \text{ s.t. } \phi(i) = i' \text{ } \forall i \in [n] \}$. Interestingly, our lower bound does not seem to apply to CCE; as we shall see in the sequel, employing certain non-external deviations is crucial for the argument.
This brings us to the notion of a \emph{sparse} distribution, introduced below.

\begin{definition}[Sparse distribution]
    \label{def:sparse}
    We say that a (correlated) distribution $\vmu$ on $[n] \times [n]$ is \emph{uniform $T$-sparse} if there exist $\vx^{(1)}, \dots, \vx^{(T)} \in \Delta^n$ and $\vy^{(1)}, \dots, \vy^{(T)} \in \Delta^n$ such that $\vmu = \frac{1}{T} \sum_{t=1}^T \vx^{(t)} \otimes \vy^{(t)}$.
\end{definition}
(Above, we denote by $\vx \otimes \vy \defeq \vx \vy^\top$ the outer (tensor) product of $\vx$ and $\vy$.) In words, a sparse distribution is one that can be expressed as a uniform mixture of $T \in \N$ product distributions. As we explained earlier, a \emph{Nash equilibrium} can be equivalently thought of as a $1$-sparse CE. A key connection that motivates~\Cref{def:sparse} is that $T$ iterations of no-regret learning produces, by definition, a uniform $T$-sparse distribution whose CE gap grows with the players' \emph{$\Phi$-regret} (\Cref{prop:regret-CE}).

\paragraph{Social welfare} The (expected) social welfare of a correlated distribution $\vmu$ on $[n] \times [n]$ is defined as $\sw(\vmu) \defeq \E_{(i,j) \sim \vmu} [ \mat{R}_{i, j} + \mat{C}_{i, j}]$; under a product distribution induced by $(\vx, \vy)$, we will write $\sw(\vx \otimes \vy) = \langle \vx, (\mat{R} + \mat{C}) \vy \rangle$.

\paragraph{SoS and pseudo-equilibria} The \emph{sum-of-squares (SoS)} hierarchy is a sequence of increasingly tighter semidefinite programs (SDPs), parameterized by a degree $d \in \N$, for the solution of a system of polynomial inequalities. Central to the SoS framework is the notion of a \emph{pseudo-distribution}, which is a generalization of the usual notion of a probability distribution.

\begin{restatable}[Pseudo-distribution]{definition}{pseudoexp}
    A \emph{degree-$d$ pseudo-distribution} is a discrete signed measure $\tilmu$ on $\R^m$ such that the associated linear functional (namely, the \emph{pseudo-expectation}) $\tilE_{\tilmu} : f \mapsto \tilE_{\tilmu}[f] = \sum_{ \vz : \tilmu(\vz) \neq 0} \tilmu(\vz) f(\vz)$, where $f: \R^m \to \R$, has the following properties:
    \begin{enumerate}
        \item \emph{normalization}: $\tilE_{\tilmu}[1] = 1$, and
        \item \emph{positivity}: $\tilE_{\tilmu}[p^2] \geq 0$ for every polynomial $p$ on $\R^m$ with degree at most $d/2$.
    \end{enumerate}
\end{restatable}

In this context, following~\citet{Kothari18:Sum}, we introduce a relaxation of uniform $T$-sparse CE based on the notion of a constrained pseudo-distribution (\Cref{def:contrainsed-pseudo}), which we refer to as \emph{uniform $T$-sparse pseudo-CE} (a $1$-sparse pseudo-CE will also be called \emph{pseudo-NE}).

\begin{definition}
    \label{def:pseudo-CE}
    For a game $(\mat{R}, \mat{C})$, a degree-$d$, uniform $T$-sparse \emph{pseudo-CE} is a degree-$d$ pseudo-distribution on $(\vx^{(1)}, \dots, \vx^{(T)}, \vy^{(1)}, \dots, \vy^{(T)})$ that satisfies $\vx^{(1)}, \dots, \vx^{(T)}, \vy^{(1)}, \dots, \vy^{(T)} \in \Delta^n$ and for all $(\phi_x, \phi_y) \in \Phi \times \Phi$ the system of quadratic inequalities
    \begin{equation}
        \label{eq:CE-prod}
        \frac{1}{T} \sum_{t=1}^T \langle \vx^{(t)}, \mat{R} \vy^{(t)} \rangle \geq \frac{1}{T} \sum_{t=1}^T \langle \phi_x(\vx^{(t)}), \mat{R} \vy^{(t)} \rangle \text{ and } \frac{1}{T} \sum_{t=1}^T \langle \vx^{(t)}, \mat{C} \vy^{(t)} \rangle \geq \frac{1}{T} \sum_{t=1}^T \langle \vx^{(t)}, \mat{C} \phi_y(\vy^{(t)}) \rangle.
    \end{equation}
\end{definition}
(We clarify that $\phi_x, \phi_y : \Delta^n \to \Delta^n$ are linear functions, and so~\eqref{eq:CE-prod} above is indeed quadratic.) \Cref{def:pseudo-CE} constitutes the natural SoS relaxation for computing uniform $T$-sparse CE. An \emph{$\epsilon$-pseudo-CE} incorporates an additive slackness $\epsilon > 0$ in~\eqref{eq:CE-prod}.

\paragraph{Oblivious rounding with a verification oracle}

We first state the definition of an oblivious rounding algorithm~\citep{Feige16:Oblivious}.

\begin{definition}[Oblivious rounding algorithm]
    A degree-$d$ \emph{oblivious rounding algorithm} for a game $\cG$ takes as input a degree-$d$, uniform $T$-sparse $\epsilon$-pseudo-CE per~\Cref{def:pseudo-CE}, and has to output a uniform $T$-sparse $\epsilon$-CE of $\cG$.
\end{definition}

As explained earlier, it is desirable to also capture rounding algorithms endowed with the ability to adaptively produce a list of candidate solutions by checking whether one of them is indeed a solution through a \emph{verification oracle}.

\begin{definition}[Verification oracle]
    \label{def:VO}
    A \emph{verification oracle} with respect to a game $\cG$ takes as input a candidate solution, in the form of $(\vx^{(1)}, \dots, \vx^{(T)})$ and $(\vy^{(1)}, \dots, \vy^{(T)})$, and outputs $\accept$ if $\vmu = \frac{1}{T} \sum_{t=1}^T \vx^{(t)} \otimes \vy^{(t)}$ is an $\epsilon$-CE of $\cG$ and $\reject$ otherwise.
\end{definition}

We are now ready to introduce the general class of rounding algorithms captured by the upcoming lower bounds.

\begin{definition}
    A \emph{degree-$d$, $q$-query oblivious rounding algorithm} with \emph{verification oracle} ($\OV$ rounding algorithm) is a degree-$d$ oblivious rounding algorithm that can additionally access a verification oracle for the underlying game at most $q$ times.
\end{definition}

Lower bounds within this algorithmic framework boil down to proving that any $\OV$ rounding algorithm either requires high degree or must otherwise submit a large number of (potentially adaptive) queries.

\subsection{Our results}
\label{sec:results}

As expected based on existing upper bounds, the lower bounds we obtain with respect to the sparsity parameter of~\Cref{def:sparse} crucially depend on the desired precision of the approximation (per~\Cref{def:CE}). In particular, we present results on two different regimes: $\epsilon = \poly(1/n)$ and $\epsilon = \poly(1/\log n)$; we shall refer to those as \emph{high} and \emph{low precision}, respectively. We begin with the high-precision regime. Here, our main result is summarized below.

\begin{theorem}
    \label{theorem:main-high}
    Suppose that there is a degree-$d$, $q$-query $\OV$ rounding algorithm for uniform $n^{1 - o(1)}$-sparse $\epsilon$-CE, with $\epsilon = n^{-c}$ for some constant $c$. Then, either $d = 2^{ \Omega( \sqrt{ \log n \log \log n } )}$ or $q = 2^{\Omega(n)}$.
\end{theorem}

It is worth noting here that there is a tradeoff between the sparsity and the degree precluded by~\Cref{theorem:main-high}; one can elevate the degree at the cost of reducing the sparsity---the extreme case being the result of~\citet{Kothari18:Sum} in which $T = 1$ and $d = \Omega(n)$. To put~\Cref{theorem:main-high} into better context, we point out that an exact CE can be computed in polynomial time via a linear program, and every distribution can be expressed as a mixture---albeit not necessarily uniform per~\Cref{def:sparse}---of $n$ product distributions; that is, modulo the use of non-uniform mixtures, the sparsity ruled out by~\Cref{theorem:main-high} is the best one can hope for.

Turning to the low-precision regime, our result mirrors~\Cref{theorem:main-high} but with quantitatively weaker bounds on $d$ and $q$.

\begin{theorem}
    \label{theorem:main-low}
    Suppose that there is a degree-$d$, $q$-query $\OV$ rounding algorithm for uniform $\log n$-sparse $\epsilon$-CE, with $\epsilon = (\log n)^{-c}$ for some constant $c$. Then, either $d = \Omega(\frac{\log n}{\log \log n}) $ or $q = n^{\Omega(\log n)}$.
\end{theorem}

A few remarks are again in order. First, even when $T = 1$, there is a quasipolynomial-time algorithm due to~\citet{Lipton03:Playing} based on exhaustive enumeration on a carefully constructed search space. Further, since $\Phi$ contains a polynomial number of deviations, there are (efficient) no-regret dynamics producing an $O(\log n/\epsilon^2)$-sparse CE (\emph{e.g.}, via the algorithm of~\citet{Gordon08:No}). In particular, in the regime of~\Cref{theorem:main-low}, there is a polynomial-time algorithm for computing $\log^{O(1)} n$-sparse CE, thereby matching the sparsity precluded by~\Cref{theorem:main-low} up to a constant in the exponent; as in~\Cref{theorem:main-high}, optimizing the dependence on $\epsilon$ (as a function of $n$) was not our focus. In fact, by virtue of recent breakthrough results~\citep{Dagan24:From,Peng24:Fast}, it is worth noting that polylogarithmic sparsity is also attainable with respect to the entire set of swap deviations when $\epsilon = \Theta(1)$.

\subsection{Technical overview}

The proof of~\Cref{theorem:main-high,theorem:main-low} follows the blueprint of~\citet{Kothari18:Sum} but with certain important twists and modifications, which are the subject of this subsection. The overarching goal is to establish~\Cref{theorem:game-collection-high,theorem:game-collection-low}, which we state and discuss below.

\begin{theorem}
    \label{theorem:game-collection-high}
    Let $T = n^{1 - o(1)}$. There is a family of $2^{\Omega(n)}$ games $\{ \mathcal{G}_k \}_{k}$ such that
    \begin{enumerate}
        \item there is a degree-$2^{ \Omega( \sqrt{ \log n \log \log n } )}$, uniform $T$-sparse pseudo-CE for every $\cG_k$ simultaneously,\label{item:property1}
        \item for every $k \neq k'$, if $\vmu$ and $\vmu'$ are uniform $T$-sparse $\poly(1/n)$-CE of $\mathcal{G}_k$ and $\mathcal{G}_{k'}$, $\vmu \neq \vmu'$.\label{item:property2}
    \end{enumerate}
\end{theorem}

\Cref{theorem:game-collection-high} immediately implies~\Cref{theorem:main-high}. Indeed, suppose that one selects a game from the family described in~\Cref{theorem:game-collection-high}, unbeknownst to the underlying algorithm. By the first property (\Cref{item:property1}), a degree-$2^{ \Omega( \sqrt{ \log n \log \log n } )}$, uniform $T$-sparse pseudo-CE could not provide any information to discern the game from the class. An $\OV$ rounding algorithm is thus left with only query access to the verification oracle. But, by the second property (\Cref{item:property2}), each query can eliminate only a single game at a time, thereby leading to~\Cref{theorem:main-high}. The counterpart of~\Cref{theorem:game-collection-high} in the low-precision regime is stated below.

\begin{theorem}
    \label{theorem:game-collection-low}
    Let $T = \log n$. There is a family of $n^{\Omega(\log n)}$ games $\{ \mathcal{G}_k \}_{k}$ such that
    \begin{enumerate}
        \item there is a degree-$\Omega(\frac{\log n}{\log \log n})$, uniform $T$-sparse pseudo-CE for every $\cG_k$ simultaneously,
        \item for every $k \neq k'$, if $\vmu$ and $\vmu'$ are uniform $T$-sparse $\poly(1/\log n)$-CE of $\mathcal{G}_k$ and $\mathcal{G}_{k'}$, $\vmu \neq \vmu'$.
    \end{enumerate}
\end{theorem}

Establishing~\Cref{theorem:game-collection-high,theorem:game-collection-low} can be in turn divided into three key steps, which are treated separately in what follows.

\paragraph{SoS hardness for the welfare maximization variant} The first one revolves around establishing SoS lower bounds for uniform $T$-sparse CE but under an additional welfare maximization constraint; of course, the final problem of interest does not involve any welfare constraint, but this focus will be justified in the next part of the argument. Starting from the high-precision regime, let us first consider the special case where $T = 1$ (that is, Nash equilibria). With a slight modification to the construction of~\citet{Gilboa89:Nash}, \citet{Kothari18:Sum} observed that there is a reduction from (the decision variant of) the independent set problem to deciding whether there is a Nash equilibrium exceeding a certain welfare threshold. In conjunction with the SoS hardness result for the independent set problem due to~\citet{Tulsiani09:CSP}, they were able to construct a game such that there is a pseudo-NE exceeding a certain welfare threshold, while in reality all NE attain welfare considerably below it; such a game is referred to as $\soshard$. For the more challenging problem of uniform $T$-sparse CE (still subject to maximizing welfare), the question is to understand how the sparsity constraint interacts with such reductions.

We show that uniform $T$-sparse CE can still be accounted for by starting from a \emph{gap-amplified} (PCP-type) version of independent set. In particular, our starting point is a recent lower bound for welfare-optimal sparse \emph{coarse} correlated equilibria~\citep{Anonymous24:Barriers}. Even though a uniform $T$-sparse CE is a stronger notion than uniform $T$-sparse CCE, hardness results for the latter do not readily carry over because of the underlying welfare constraint. Indeed, it turns out that in their reduction any CE obtains lower welfare than the welfare-optimal CCE. As a result, we observe that a direct adaptation of their approach to sparse CE necessitates a gap of $\Theta(T^2)$ (\Cref{lemma:IS}). One of our key contributions is to appropriately leverage internal deviations so as to reduce the gap of the reduction to $\Theta(T)$ (\Cref{lemma:ref-IS}); our reduction for sparse CE presents several interesting geometric features, discussed in detail in~\Cref{sec:hard-sos}.

Now, crucially, the result of~\citet{Tulsiani09:CSP} can support a gap up to $n^{1 - o(1)}$ (\Cref{theorem:Tulsiani-IS}), and this is precisely why the sparsity bound $T = n^{1 - o(1)}$ appears in~\Cref{theorem:game-collection-high}. In particular, this leads to an $\soshard$ game with respect to uniform $T$-sparse CE (per~\Cref{def:soshard}), as we formalize in~\Cref{theorem:soshard-high}. The argument in the low-precision regime has a similar flavor (\Cref{theorem:soshard-low}), but instead relies on SoS lower bounds for the \emph{planted clique} problem~\citep{Pang21:SOS}. It is worth noting that~\citet{Kothari18:Sum} had to follow a different path since Pang's result was not available at the time---and earlier integrality gaps~\citep{Barak19:Nearly} were not quite suited for such purposes.

\paragraph{Hardness for enumeration algorithms} The second step is to provide (information-theoretic) lower bounds against algorithms that access the game only through a verification oracle---\emph{oblivious algorithms} in the parlance of~\citet{Daskalakis09:Oblivious}; the rationale behind this consideration is that an $\OV$ rounding algorithm that obtains no meaningful information from the SoS relaxation is thereby reduced to an oblivious algorithm. This problem is well-understood for Nash equilibria since the work of~\citet{Daskalakis09:Oblivious}. The basic idea is that one can construct a large family of games whose set of $\epsilon$-Nash equilibria are pairwise disjoint (as in~\Cref{item:property2} we saw earlier); deriving sharp bounds requires a careful construction, discussed in~\Cref{sec:enum}. In \Cref{theorem:enumhard-low,theorem:enumhard-high}, we show that such bounds carry over to coarse correlated equilibria as well (no matter the sparsity). The simple observation here is that the existing lower bounds for Nash equilibria against oblivious algorithms are based on (modulo strictly dominated actions that can be easily accounted for) constant-sum games; in such games, CCE and NE are tantamount---in the precise sense of~\Cref{fact:zerosum}. We call such a family $\enumhard$ games (\Cref{def:enumhard}).

\paragraph{Combining the games} The final crucial piece in the construction shows how to appropriately combine such games to arrive at~\Cref{theorem:game-collection-high,theorem:game-collection-low}. In particular, let $(\mat{R}, \mat{C})$ be $\soshard$: a game such that there is a degree-$d$, uniform $T$-sparse pseudo-CE in which both players obtain utility at least $\delta$, while all uniform $T$-sparse CE attain welfare at most $2 \delta - 2 \epsilon$. Consider further a family of $\enumhard$ games $\{ (\mat{R}^S, \mat{C}^S) \}_{S \in \mathcal{S}}$. The idea is to consider the game
\begin{equation*}
    \mat{R}' =
    \begin{pmatrix}
        \mat{R} & - k \mat{1}_{n \times n} \\
        \delta \mat{1}_{n \times n} & \mat{R}^{S}
    \end{pmatrix} \quad
    \text{and} \quad
    \mat{C}' =
    \begin{pmatrix}
        \mat{C} & \delta \mat{1}_{n \times n} \\
        -k \mat{1}_{n \times n} & \mat{C}^S
    \end{pmatrix},
\end{equation*}
where $\mat{1}_{n \times n}$ denotes the all-ones $n \times n$ matrix and $k \gg 1$.

We first want to argue that uniform $T$-sparse CE in $(\mat{R}', \mat{C}')$ are close to uniform $T$-sparse CE in $(\mat{R}^S, \mat{C}^S)$. 
Indeed, the crucial role of the $-k \mat{1}_{n \times n}$ term in the off-diagonal of $\mat{R}' + \mat{C}'$ is that it forces any product distribution comprising a $T$-sparse (C)CE in $(\mat{R}', \mat{C}')$ to be either supported only on $(\mat{R}, \mat{C})$ or only $(\mat{R}^S, \mat{C}^S)$---up to some probability mass that decays when $k$ is large (\Cref{claim:off-diagonal}). The goal is to prove that actually only the second case can arise: each product is essentially only supported on $(\mat{R}^S, \mat{C}^S)$, at which point it readily follows that a uniform $T$-sparse CE in $(\mat{R}', \mat{C}')$ must be a uniform $T$-sparse CE in $(\mat{R}^S, \mat{C}^S)$ as well (with roughly the same solution quality). To prove the former assertion, we show that in the contrary case, the conditional distribution on $(\mat{R}, \mat{C})$ must be a uniform $T$-sparse CE for that game (\Cref{claim:induced-CE}). This argument carefully hinges on using a richer set of deviations than external ones. The basic reason is that we want to apply a different mapping when on $(\mat{R}, \mat{C})$ compared to when being on $(\mat{R}^S, \mat{C}^S)$, a functionality not supported by external deviations. To explain this point further, an external deviation for Player $x$ to an action in $\mat{R}$ could result in an overwhelmingly negative utility due to the $-k \mat{1}_{n \times n}$ term: for every component supported on $(\mat{R}^S, \mat{C}^S)$, such a deviation inevitable incurs $-k$. It is thus unclear how to argue about deviations in $(\mat{R}, \mat{C})$ starting from $(\mat{R}', \mat{C}')$. This is no longer a concern when players are able to deviate differently on different components of the game. Continuing from~\Cref{claim:induced-CE}, we now reach a contradiction: uniform $T$-sparse CE of $(\mat{R}, \mat{C})$ attain welfare considerably lower than $2 \delta$ (by the property of the $\soshard$ game), and at the same time, there is a suitable deviation that takes advantage of the $\delta \mat{1}_{n \times n}$ term in $\mat{R}'$ and $\mat{C}'$ (\Cref{claim:Lsmall}); this contradicts the assumption that we started from a CE.

What remains to show is that there is a degree-$d$, uniform $T$-sparse pseudo-CE shared among the games no matter the selection of $(\mat{R}^S, \mat{C}^S)$ from the family of $\enumhard$ games. This is a fairly straightforward exercise, similar to the completeness proof for Nash equilibria (\Cref{lemma:stitch-completeness}).
\subsection{Related work}
\label{sec:related}

Related to sparsity per~\Cref{def:sparse}, it is worth pointing out here a notion of sparsity which instead imposes a bound on the \emph{support}---the number of nonzero elements---of the correlated distribution $\vmu$. The latter is clearly more stringent, and has been the subject of much investigation in the past (\emph{e.g.}, \citep{Babichenko14:Simple,Feder07:Approximating,Althofer94:Sparse}). From the perspective of regret minimization, lower bounds on the support of $\vmu$ have been successfully employed so as to preclude fast convergence of no-regret dynamics when players select \emph{pure strategies}, while the more permissive notion we study here accounts for \emph{mixed strategies} as well. One crucial difference between the two is that there are no such information-theoretic barriers surrounding \Cref{def:sparse}, in that a $1$-sparse CE always exists~\citep{Nash50:Non}; this is precisely why we need to resort to computational lower bounds.

\citet{Foster23:Hardness} recently introduced~\Cref{def:sparse} in the context of Markov games (see also~\citep{Peng24:Complexity}), and showed that---under standard complexity assumptions---no polynomial sparsity can be attained in such games; this stands in stark contrast to normal-form games. Unfortunately, directly translating such lower bounds to normal-form games appears to be of little use as it does not result in a polynomial-time reduction. Yet, it is worth noting that certain lower bounds have been documented for restricted classes of algorithms, such as multiplicative weights update~\citep{Peng24:Complexity}. Besides those recent works, sparsity played a key role in the celebrated \emph{ellipsoid against hope} algorithm of~\citet{Papadimitriou08:Computing} (\emph{cf.}~\citep{Jiang11:Polynomial,Farina24:Polynomial}). Indeed, that algorithm outputs a CE that is a convex combination of polynomially many product distributions, and is able to compute a correlated equilibrium even in many succinctly represented multi-player games; a correlated distribution in such games is in general an exponential object, but sparsity crucially provides a succinct representation.

Beyond the computational perspective, a rather orthogonal approach to proving lower bounds for no-regret dynamics is to resort to query or communication complexity (\emph{e.g.}, \citep{Maiti23:Query,Conitzer04:Communication,Goldberg16:Bounds,Goldberg23:Lower,Babichenko15:Query,Fearnley16:Finding,Goos23:Near,Hadiji23:Towards}). One compelling aspect of the former compared to the latter approach is that computational lower bounds apply even in the centralized model of computation where the entire game is known in advance. This is crucial given that no-regret learning is consistently part of the best known algorithms for equilibrium computation problems. Further, query or communication lower bounds are quite brittle depending on how the game is accessed and information is distributed among the players. On the other hand, it has to be noted that the model we operate in is so permissive that no meaningful lower bounds can be established in, for example, (two-player) zero-sum games.
\section{Further preliminaries}
\label{sec:prels}

Continuing from~\Cref{sec:init-prels}, this section provides some further background on regret minimization and the SoS paradigm.

\paragraph{Notation} We let $\N = \{1, 2, \dots \}$. We use the shorthand notation $[n] \defeq \{1, \dots, n\}$. We use lowercase boldface letters, such as $\vx$ and $\vy$, to represent vectors. The $i$th coordinate of $\vx \in \R^n$ is accessed by $\vx_i$. For $S \subseteq [n]$, $\vx_S \in \R^{S}$ denotes the projection of $\vx$ on the coordinates in $S$. For $\vx, \vx' \in \R^n$, we use $\langle \vx, \vx' \rangle \defeq \sum_{i=1}^n \vx_i \vx_i'$ for their inner product. $\|\vx\|_1 \defeq \sum_{i=1}^n |\vx_i| $ denotes the $\ell_1$ norm of $\vx$ and $\| \vx\|_\infty \defeq \max_{i \leq n} |\vx_i|$ the $\ell_\infty$ norm. For $\vx, \vx' \in \R^n$ and $S \subseteq [n]$, we use the convention that $\|\vx - \vx'_S \|_1 = \|\vx_S - \vx_S'\|_1 + \sum_{i \notin S} |\vx_i|$. $\vec{e}_i \in \Delta^n$ is defined such that its $i$th entry is equal to $1$. If $\vx \in \R^{n + m}$, for $n, m \in \N$, we let $\R^n \ni \vx_{i \leq n} \defeq (\vx_1, \dots, \vx_n)$ and $\R^m \ni \vx_{i \geq n+1} \defeq (\vx_{n+1}, \dots, \vx_{n+m})$. For matrices, we use uppercase boldface letter, such as $\mat{R}$ and $\mat{C}$. The $(i,j)$ entry of $\mat{R}$ is accessed by $\mat{R}_{i, j}$. $\mat{I}_n$ is the $n \times n$ identity matrix, $\mat{0}_{n \times n}$ is the all-zeroes $n \times n$ matrix, while $\mat{1}_{n \times n}$ is the all-ones $n \times n$ matrix. Superscripts (with parantheses) are typically used for quantities related to the discrete time. We use the $O(\cdot), \Theta(\cdot), \Omega(\cdot)$ notation to suppress absolute constants.

\subsection{Correlated equilibria and regret minimization}
\label{section:CE-Regret}

We begin by introducing precisely the notion of a correlated equilibrium we employ and its connection with regret minimization. For further background, we refer to the excellent book of~\citet{Cesa-Bianchi06:Prediction}.

\paragraph{Correlated equilibria} The concept of the \emph{correlated equilibrium}~\citep{Aumann74:Subjectivity} was presented earlier (as \Cref{def:CE}), but we deferred the definition of set $\Phi$ until now. To do so, we let $\Phi_{\textrm{int}}$ be the set of \emph{internal deviations}~\citep{Stoltz05:Internal,Blum07:From} on $[n] \to [n]$; namely, $\phi \in \Phi_{\textrm{int}}$ if there exists $i' \in [n]$ such that $\phi(i) = i$ for all $i \in [n] \setminus \{i' \}$. $\Phi_{\textrm{int}}$ is a subset of $\Phi_{\textrm{swap}}$---the latter being the set of all deviations from $[n]$ to $[n]$. In the literature, correlated equilibria are typically defined with respect to either $\Phi_{\textrm{int}}$ or $\Phi_{\textrm{swap}}$ (\emph{e.g.}, \citep{Ganor18:Communication,Goldberg16:Bounds}). Up to a factor of $n$, those two definitions are equivalent~\citep{Blum07:From}. Nevertheless, especially in the low-precision regime, those two should not be blurred. In this context, our lower bounds do not require the entire set of swap deviations; instead, it suffices to consider the subset of swap deviations for which there exist $i', i_0 \in [n]$ such that
\[
    \phi(i) = 
    \begin{cases}
        i' & i \leq i_0,\\
        i & i > i_0.
    \end{cases}
\]
This set, say $\widehat{\Phi}$, contains as a subset the set of external deviations (by taking $i_0 = n$). It is a somewhat artificial set, and so for more clarity in the exposition, we augment that set with internal deviations as well. That is, we work with $\Phi \defeq \widehat{\Phi} \cup \Phi_{\textrm{int}}$. (Alternatively, we could have worked with $\Phi_{\textrm{swap}}$ throughout, but we prefer to derive our lower bounds with respect to a set with a polynomial number of deviations.)

For a mixed strategy $\vx \in \Delta^n$, we use (for example, in~\eqref{eq:CE-prod} earlier) the extended mapping (with an overload in the notation) $\Phi \ni \phi_x : \vx \mapsto \E_{i \sim \vx} \phi_x(i)$. In particular, it is useful to observe that $\phi_x$ can be represented through a (column-)stochastic matrix $\mat{M}_x$, so that $\phi_x(\vx) = \mat{M}_x \vx$ for all $\vx \in \Delta^n$.

\paragraph{Connection with no-regret learning} While the notion of a sparse CE (\Cref{def:sparse}) is natural in its own right, our motivation stems from its intimate connection with the framework of \emph{no-regret learning}. The upshot is that $T$ rounds of learning produces, essentially by definition, a uniform $T$-sparse correlated distribution whose CE gap (\Cref{def:CE}) can be bounded in terms of the players' $\Phi$-regret (\Cref{prop:regret-CE}).

Taking a step back, the framework of no-regret learning sees a learner interact in a sequence of $T$ iterations with an environment. In each round, the learner selects a strategy $\vx^{(t)} \in \Delta^n$, whereupon the environment selects a (linear) utility function $\vx \mapsto \langle \vx, \vu^{(t)} \rangle$; at that round, the learner obtains a utility of $\langle \vx^{(t)}, \vu^{(t)} \rangle$, and observes $\vu^{(t)}$ as feedback from the environment. How should the learner engage in this repeated interaction? A canonical measure of performance revolves around \emph{$\Phi$-regret}~\citep{Stoltz07:Learning,Greenwald03:Correlated}, defined as
\begin{equation}
    \label{eq:Phireg}
    \reg_{\Phi}^T \defeq \max_{\phi \in \Phi} \left\{ \sum_{t=1}^T \langle \phi(\vx^{(t)}), \vu^{(t)} \rangle \right\} - \sum_{t=1}^T \langle \vx^{(t)}, \vu^{(t)} \rangle.
\end{equation}

For example, \eqref{eq:Phireg} under $\Phi_{\text{ext}}$ gives rise to external regret, while $\Phi_{\text{int}}$ to internal regret. At this point, it should be evident that there is an immediate connection between sparse correlated equilibria and $\Phi$-regret, as specified below.

\begin{proposition}
    \label{prop:regret-CE}
    Consider $T$ iterations in which Player $x$ observes $\vu_x^{(t)} = \mat{R} \vy^{(t)}$ and Player $y$ observes $\vu_y^{(t)} = \mat{C}^\top \vx^{(t)}$ for all $t \in [T]$. Then, if the maximum of the players' $\Phi$-regret is $\reg_{\Phi}^T$, the correlated distribution $\frac{1}{T} \sum_{t=1}^T \vx^{(t)} \otimes \vy^{(t)}$ is a $\reg_{\Phi}^T / T$-CE.
\end{proposition}

\subsection{Sum-of-squares (SoS)}

We continue by stating some basic facts about SoS algorithms. It might first be helpful to restate the definition of a pseudo-expectation.

\pseudoexp*

The point here is that one relaxes the constraint $\tilmu \geq 0$ of a measure by only requiring that $\tilmu$ satisfies low-degree nonnegativity constraints. Indeed, a degree-$\infty$ pseudo-distribution is a probability distribution on $\R^m$. Since $\tilE_{\tilmu}$ is a linear operator, it can be fully specified by the order-$d$ tensor $\tilE_{\tilmu} [ (1, \vz)^{\otimes d}]$, which contains the \emph{pseudo-moments} of $\tilmu$. A crucial property that drives the SoS framework is that there is a separation oracle with running time $m^{O(d)}$ for the convex set of the degree-$d$ moment tensors of all degree-$d$ pseudo-distributions.

\begin{fact}[\citep{Nesterov00:Semidefinite,Parrilo00:Structured,Lasserre01:Global}]
    For any $d \in \N$, the convex set
    \[
        \left\{ \tilE_{\tilmu} (1, \vz)^{\otimes d} : \text{ $\tilmu$ is a degree-$d$ pseudo-distribution on $\R^m$}\right\}
    \]
    admits a weak separation~\citep{Grotschel81:Ellipsoid} with running time $m^{O(d)}$.
\end{fact}

This fact enables efficiently optimizing with respect to the set of pseudo-distributions. Now, let us consider a system of $\ell$ polynomial inequalities $\{ f_1 \geq 0, \dots, f_{\ell} \geq 0 \}$ on $\R^m$. A pseudo-distribution $\tilmu$ satisfies a polynomial inequality constraint $f \geq 0$ if for every polynomial $p$ with $\deg(p^2 f) \leq d$, it holds that $\tilE_{\tilmu}[p^2 f] \geq 0$. More broadly, we have the following definition (to which~\Cref{def:pseudo-CE} is subjected).

\begin{definition}[Constrained pseudo-distribution]
    \label{def:contrainsed-pseudo}
    Let $\tilmu : \R^m \to \R $ be a degree-$d$ pseudo-distribution. We say that $\tilmu$ satisfies a system of polynomial constraints $\{ f_1 \geq 0, \dots, f_{\ell} \geq 0 \}$ if for every set $S \subseteq [\ell]$ and every sum-of-squares polynomial $h$ such that $\deg(h) + \sum_{i \in S} \deg(f_i) \leq d$,
    \[
        \tilE_{\tilmu} \left[ h \prod_{i \in S} f_i \right] \geq 0.
    \]
\end{definition}
For additional flexibility, some authors replace $\deg(f_i)$ by $\max ( r, \deg(f_i))$, in which case $\tilmu$ is said to satisfy the system at level $r \in \N$, but that will not be useful for our purposes here. Further, we say that a system of polynomial constraints is \emph{explicitly bounded} if it contains a constraint of the form $\|\vz\|^2 \leq M$. (In our applications, explicit boundedness is always met.)

\begin{fact}
    For any explicitly bounded and satisfiable system of $\ell$ polynomial constraints in $m$ variables, there is an $(m + \ell)^{O(d)}$-time algorithm that outputs a degree-$d$ pseudo-distribution satisfying that system.
\end{fact}

As a result, leveraging the above fact, it is possible to compute low-degree uniform $T$-sparse pseudo-CE (\Cref{def:pseudo-CE}) in polynomial time.

For further background on the SoS paradigm, we refer to the excellent book of~\citet{Fleming19:Semialgebraic} and the introductory lecture notes compiled by~\citet{Barak16:Proofs}.

\section{Hard games for welfare maximization with SoS}
\label{sec:hard-sos}

We begin by formally stating the definition of an $\soshard$ game. As we explained earlier, although we are interested in proving lower bounds for the unrestricted version of the problem, here we incorporate the additional constraint of maximizing welfare, which will be justified in the construction of~\Cref{sec:stitching}.

\begin{definition}
    \label{def:soshard}
    A two-player game with utilities in $[-1, 1]$ is called $(d, T, \epsilon_0, \epsilon)$-$\soshard$ if
    \begin{itemize}
        \item there is a degree-$d$ pseudo-distribution on $(\vx^{(1)}, \dots, \vx^{(T)}, \vy^{(1)}, \dots, \vy^{(T)})$, given in~\Cref{def:tilmu}, that satisfies the constraints of a $T$-sparse pseudo-CE together with the constraint that both players receive utility at least $\delta \in (0, \frac{1}{2})$; and
        \item any uniform $T$-sparse $\epsilon_0$-CE attains welfare at most $2 \delta - 2 \epsilon$.
    \end{itemize}
\end{definition}

\subsection{High precision}
\label{sec:wel-highprec}

We first treat the regime in which $\epsilon = \poly(1/n)$. Here, following~\citet{Kothari18:Sum}, we construct an $\soshard$ game via a reduction from the problem of computing an \emph{independent set} of a graph; an SoS lower bound for the latter problem was shown by~\citet{Tulsiani09:CSP} (\Cref{theorem:Tulsiani-IS}). It is worth noting that the following construction is quite different from the one devised by~\citet{Kothari18:Sum}, even when $T = 1$.

In this context, suppose that we are given an (undirected) $n$-node graph $G = ([n], E)$. The problem is to determine whether $G$ contains an independent set of size $k \in \N$, where $k$ is also given as part of the input. We recall that a subset of the nodes $S \subseteq [n]$ is \emph{independent} if for any $i, j \in S$, $\{i, j\} \notin E$; of course, the reduction below can be equivalently phrased in terms of cliques (by replacing $G$ by its complement), but we use the language of independent sets to be consistent with~\citet{Tulsiani09:CSP}. Let $\mat{A}$ denote the adjacency matrix of $G$; namely, $\mat{A}_{i, j} = 1$ if $\{i, j\} \in E$, and $\mat{A}_{i, j} = 0$ otherwise; by convention, we assume that the diagonal elements of $\mat{A}$ are equal to $1$. Then, we construct the following $(2n) \times (2n)$ two-player game:
\begin{equation}
    \label{eq:indep-setgame}
    \mat{R} \defeq \frac{1}{2} \begin{pmatrix}
        \mat{1}_{n \times n} - \mat{A} + (\gamma + 1) \mat{I}_n & - k \mat{I}_n \\
        k \mat{I}_n & \mat{0}_{n \times n}
    \end{pmatrix}
    \quad \text{and} \quad
    \mat{C} \defeq \frac{1}{2} \begin{pmatrix}
        \mat{1}_{n \times n} - \mat{A} + (\gamma + 1) \mat{I}_n & k \mat{I}_n \\
        - k \mat{I}_n & \mat{0}_{n \times n}
    \end{pmatrix}.
\end{equation}
We will eventually have to rescale the utilities of~\eqref{eq:indep-setgame} proportionally to $1/k$ to ensure that the entries of $\mat{R}$ and $\mat{C}$ are in $[-1, 1]$. The first property of $(\mat{R}, \mat{C})$ is that an independent set in $G$ with size $k$ induces a Nash equilibrium with high welfare, as specified below. (We clarify that $\vec{u}(S) \in \Delta^n$ denotes the uniform distribution over $S \subseteq [n]$.)

\begin{lemma}
    \label{lemma:completeness}
    If $G$ contains an independent set $S \subseteq [n]$ with size $k$, then the pair of strategies $( ( \vu(S), \vec{0}_n), ( \vu({S}), \vec{0}_n ) )$ is a Nash equilibrium of $(\mat{R}, \mat{C})$ with welfare $ 2 \delta \defeq 1 + \gamma/k$.
\end{lemma}

\begin{proof}
    The fact that each player obtains utility $\delta$ under the above pair of strategies follows directly from the definition of $(\mat{R}, \mat{C})$ since $\mat{A}_{i, j} = 0$ for any pair $i, j \in S$ with $i \neq j$. Further, each possible deviation from either Player $x$ or Player $y$ is easily seen to give utility at most $\delta$, concluding the proof.
\end{proof}

We will use~\Cref{lemma:completeness} later in conjunction with a basic fact on SoS reductions to construct a pseudo-NE that ``pretends'' to enjoy a welfare of $2 \delta$. First, we turn to the soundness of the reduction, which uses ideas from the recent lower bound concerning welfare maximization with sparse CCE~\citep{Anonymous24:Barriers}; as we shall see, obtaining tight lower bounds for sparse CE requires some further ideas.

The first observation is that in $(\mat{R}, \mat{C})$, with the exception of the (top-left) $[n] \times [n]$ component, the rest of the game is zero-sum. This means that a (C)CE with high welfare cannot afford to allocate too much probability mass outside of that component, as formalized below.

\begin{lemma}
    \label{lemma:bad-prob}
    Let $\vmu = \frac{1}{T} \sum_{t=1}^T \vx^{(t)} \otimes \vy^{(t)}$ be an $\epsilon$-CE in $(\mat{R}, \mat{C})$ with welfare at least $1$. Then, if $\hvx^{(t)}_{i \leq n} \defeq \nicefrac{\vx_{i \leq n}^{(t)}}{\sum_{i=1}^n \vx_i^{(t)}}$ and $\hvy^{(t)}_{j \leq n} \defeq \nicefrac{\vy_{j \leq n}^{(t)}}{\sum_{j=1}^n \vy_j^{(t)}}$, $\hvmu \defeq \frac{1}{T} \sum_{t=1}^T \hvx^{(t)} \otimes \hvy^{(t)}$ is an $(\epsilon + 4 \gamma k)$-CE of $(\mat{R}, \mat{C})$.
\end{lemma}

\begin{proof}
    We define $\delta^{(t)} \defeq 1 - \sum_{i=1}^n \sum_{j=1}^n \vx_i^{(t)} \vy_j^{(t)}$. By construction of $(\mat{R}, \mat{C})$, we have
    \[
        1 \leq \sw(\vmu) \leq (1 + \gamma) \pr_{(i, j) \sim \vmu} [i \leq n, j \leq n] = (1 + \gamma) \frac{1}{T} \sum_{t=1}^T  \sum_{i=1}^n \vx_i^{(t)} \sum_{j=1}^n \vy_j^{(t)}
    \]
    since the maximum utility in $\mat{R} + \mat{C}$ is $1 + \gamma$ and $\mat{R}_{i, j} + \mat{C}_{i, j} = 0$ when $i > n$ or $j > n$. Thus,
    \[
        \frac{1}{T} \sum_{t=1}^T \left(1 - \sum_{i=1}^n \sum_{j=1}^n \vx_i^{(t)} \vy_j^{(t)} \right) \leq \frac{\gamma}{1 + \gamma} \leq \gamma.
    \]
    That is, $\frac{1}{T} \sum_{t=1}^T \delta^{(t)} \leq \gamma$. To prove that $\hvmu$ is an $(\epsilon + 4 \gamma k)$-CE, we will bound deviations of Player $x$, and then the claim will follow by symmetry. Since $\vmu$ is assumed to be an $\epsilon$-CE in $(\mat{R}, \mat{C})$, we have
    \begin{equation}
        \label{eq:init-bound}
        \frac{1}{T} \sum_{t=1}^T \langle \phi_x(\vx^{(t)}) - \vx^{(t)}, \mat{R} \vy^{(t)} \rangle \leq \epsilon
    \end{equation}
    for any $\phi_x \in \Phi$. Let us fix some $\phi_x \in \Phi$. First, we bound
    \begin{align}
        \langle \vx^{(t)}, \mat{R} \vy^{(t)} \rangle &= \langle \vx^{(t)} - \hvx^{(t)}, \mat{R} \vy^{(t)} \rangle + \langle \hvx^{(t)}, \mat{R} (\vy^{(t)} - \hvy^{(t)}) \rangle +  \langle \hvx^{(t)}, \mat{R} \hvy^{(t)} \rangle \notag \\
        &\leq \frac{k}{2} \|\vx^{(t)} - \hvx^{(t)} \|_1 + \frac{k}{2} \|\vy^{(t)} - \hvy^{(t)} \|_1 + \langle \hvx^{(t)}, \mat{R} \hvy^{(t)} \rangle \label{align:max-util} \\
        &= k \left( 1 - \sum_{i=1}^n \vx_i^{(t)} + 1 - \sum_{j=1}^n \vy_j^{(t)} \right) + \langle \hvx^{(t)}, \mat{R} \hvy^{(t)} \rangle \notag \\
        &\leq 2 k \left( 1 - \sum_{i=1}^n \sum_{j=1}^n \vx_i^{(t)} \vy_j^{(t)} \right) + \langle \hvx^{(t)}, \mat{R} \hvy^{(t)} \rangle,\notag \\
        &= 2k \delta^{(t)} + \langle \hvx^{(t)}, \mat{R} \hvy^{(t)} \rangle,
    \end{align}
    where~\eqref{align:max-util} uses the fact that the maximum utility in $\mat{R}$ is $k/2$. So, summing over all $t \in [T]$,
    \begin{equation}
        \label{eq:bound1}
        \frac{1}{T} \sum_{t=1}^T \langle \vx^{(t)}, \mat{R} \vy^{(t)} \rangle \leq 2 k \gamma + \frac{1}{T} \sum_{t=1}^T \langle \hvx^{(t)}, \mat{R} \hvy^{(t)} \rangle.
    \end{equation}
    Similarly, we have
    \begin{align}
        \langle \phi_x(\vx^{(t)}), \mat{R} \vy^{(t)} \rangle &= \langle \phi_x(\vx^{(t)}) - \phi_x(\hvx^{(t)}), \mat{R} \vy^{(t)} \rangle + \langle \phi_x(\hvx^{(t)}), \mat{R} (\vy^{(t)} - \hvy^{(t)}) \rangle + \langle \phi_x(\hvx^{(t)}), \mat{R} \hvy^{(t)} \rangle \notag \\
        &\geq - \frac{k}{2} \| \mat{M}_x  (\vx^{(t)} - \hvx^{(t)} ) \|_1 - \frac{k}{2} \|\vy^{(t)} - \hvy^{(t)} \|_1 + \langle \phi_x(\hvx^{(t)}), \mat{R} \hvy^{(t)} \rangle \label{align:mat-rep} \\
        &\geq - \frac{k}{2} \| \vx^{(t)} - \hvx^{(t)} \|_1 - \frac{k}{2} \|\vy^{(t)} - \hvy^{(t)} \|_1 + \langle \phi_x(\hvx^{(t)}), \mat{R} \hvy^{(t)} \rangle \label{align:col-stochastic} \\
        &\geq - 2 k \delta^{(t)} + \langle \phi_x(\hvx^{(t)}), \mat{R} \hvy^{(t)} \rangle,\label{align:bound2}
    \end{align}
    where in~\eqref{align:mat-rep} we denote by $\mat{M}_x$ the matrix representation of $\phi_x$ and \eqref{align:col-stochastic} follows from the fact that $\mat{M}_x$ is column-stochastic. Summing~\eqref{align:bound2} over all $t \in [T]$,
    \begin{equation}
        \label{eq:realbound2}
        \frac{1}{T} \sum_{t=1}^T \langle \phi_x(\vx^{(t)}), \mat{R} \vy^{(t)} \rangle \geq - 2 k \gamma + \frac{1}{T} \sum_{t=1}^T \langle \phi_x(\hvx^{(t)}), \mat{R} \hvy^{(t)} \rangle.
    \end{equation}
    Thus, combining~\eqref{eq:bound1} and~\eqref{eq:realbound2} with~\eqref{eq:init-bound} concludes the proof. 
\end{proof}

Thus, by selecting $\gamma \ll 1/k$, we can essentially work under the assumption that the underlying CCE is supported solely on $[n] \times [n]$. The next lemma explains the role of the $\frac{k}{2} \mat{I}_n$ component in the off-diagonal of $(\mat{R}, \mat{C})$.

\begin{lemma}
    \label{lemma:ub-prob}
    Let $\vmu \defeq \frac{1}{T} \sum_{t=1}^T \vx^{(t)} \otimes \vy^{(t)}$ be an $\epsilon$-CCE of $(\mat{R}, \mat{C})$ supported only on $[n] \times [n]$. Then, for any $i, j \in [n]$ and $t \in [T]$, $\vx_i^{(t)}, \vy_j^{(t)} \leq (1 + \gamma + 2 \epsilon) T/k$.
\end{lemma}

\begin{proof}
    Since $\vmu$ is supported only on $[n] \times [n]$, it follows that the utility of each player under $\vmu$ is at most $(1 + \gamma)/2$. Now, if Player $y$ was to deviate to action $n + i$, the obtained utility would be at least $\frac{k}{2 T} \sum_{t=1}^T \vx_i^{(t)}$. But since $\vmu$ is an $\epsilon$-CCE, we have
    \[
        \frac{1+ \gamma}{2} \geq \frac{k}{2T} \max_{t \leq T} \vx_i^{(t)} - \epsilon,
    \]
    and the bound with respect to strategies of Player $x$ follows. The bound for Player $y$ follows symmetrically.
\end{proof}

We call attention to the fact that the lemma above applies to (approximate) CCE as well. In fact, obtaining tight lower bounds for CE will require strengthening~\Cref{lemma:ub-prob} by suitably leveraging internal deviations (formalized in~\Cref{lemma:ref-ub}). First, to better illustrate the key idea behind that refinement, we will see how~\Cref{lemma:ub-prob} already precludes a sparsity up to $\approx \sqrt{n}$. Before we proceed, we state a simple claim that will be used in the proof.

\begin{claim}
    \label{claim:basic}
    Suppose that $0 \leq \vx_i \leq \frac{1}{\ell}$ for any $i \in [n]$ and $\sum_{i=1}^n \vx_i = p$. Then, $\sum_{i=1}^n \vx_i^2 \leq \frac{p}{\ell}$.
\end{claim}

\begin{lemma}
    \label{lemma:IS}
    Suppose that there is a uniform $T$-sparse $\frac{1}{4}$-CE in $(\mat{R}, \mat{C})$ supported only on $[n] \times [n]$ with welfare at least $1 + 3\gamma/(4 k)$. Then, if $\gamma = 1/(256 k T^3)$, there is an independent set in $G$ with size $k/(16T^2)$.
\end{lemma}

\begin{proof}
    Since $\sw(\vmu) \geq 1 + 3 \gamma/(4 k)$ (by assumption), it follows that there is some $t \in [T]$ such that $\sw(\vx^{(t)} \otimes \vy^{(t)}) \geq 1 + 3 \gamma/(4 k)$. Further, by construction of $(\mat{R}, \mat{C})$,
    \[
        \sw(\vx^{(t)} \otimes \vy^{(t)}) \leq 1 + \gamma \sum_{i=1}^n \vx_i^{(t)} \vy_i^{(t)},
    \]
    which in turn implies that $\sum_{i=1}^n \vx_i^{(t)} \vy_i^{(t)} \geq 3/(4 k)$. By~\Cref{lemma:ub-prob} and our choice of parameters, it follows that $\vx_i^{(t)}, \vy_j^{(t)} \leq 2 T/k$ for any $i, j \in [n]$. We now select $S \defeq \{ i \in [n] : \vx_i^{(t)}, \vy_i^{(t)} \geq 1/(16 k T) \}$. Then,
    \begin{align}
        \frac{3}{4k} \leq \sum_{i=1}^n \vx_i^{(t)} \vy_i^{(t)} &= \sum_{i \in S } \vx_i^{(t)} \vy_i^{(t)} + \sum_{i \notin S} \vx_i^{(t)} \vy_i^{(t)} \notag \\
        &= \sum_{i \in S } \vx_i^{(t)} \vy_i^{(t)} + \sum_{i \notin S} \max( \vx_i^{(t)}, \vy_i^{(t)}) \min(\vx_i^{(t)}, \vy_i^{(t)}) \notag \\
        &\leq |S| \frac{4 T^2}{k^2} + \sqrt{ \sum_{i \notin S} \max(\vx_i^{(t)}, \vy_i^{(t)})^2 \sum_{i \notin S} \min(\vx_i^{(t)}, \vy_i^{(t)})^2} \label{align:cs-b} \\
        &\leq |S| \frac{4 T^2}{k^2} + \sqrt{ \frac{4 T}{k} \frac{1}{16 k T}} \leq |S| \frac{4 T^2}{k^2} + \frac{1}{2 k},\label{align:claimp}
    \end{align}
    where~\eqref{align:cs-b} uses the fact that $\vx_i^{(t)}, \vy_j^{(t)} \leq 2 T / k$ for any $i,j \in [n]$ and the Cauchy-Schwarz inequality, and~\eqref{align:claimp} uses~\Cref{claim:basic}; in particular, for the application of~\Cref{claim:basic}, we note that $\sum_{i \notin S} \max( \vx_i^{(t)}, \vy_i^{(t)}) \leq 2$; $\sum_{i \notin S} \min( \vx_i^{(t)}, \vy_i^{(t)}) \leq 1$; and $\min(\vx_i^{(t)}, \vy_i^{(t)}) \leq 1/(16 k T)$ by definition of $S$. Rearranging~\eqref{align:claimp} implies that $|S| \geq k/(16  T^2)$.

    We will finally argue that $S$ indeed induces an independent set in $G$. For the sake of contradiction, suppose that there are $i, j \in S$ with $i \neq j$ such that $\mat{A}_{i, j} = 1$. Then, by the property of $S$,
    \[
        \sw(\vx^{(t)} \otimes \vy^{(t)}) \leq 1 + \gamma \frac{2T}{k} - \frac{1}{128 k^2 T^2}.
    \]
    Taking $\gamma = 1/(256 k T^3)$ contradicts the fact that $\sw(\vx^{(t)} \otimes \vy^{(t)}) \geq 1 + 3 \gamma/(4k)$.
\end{proof}

As a result, given that approximating the maximum independent set to a factor of $n^{1 - o(1)}$ is intractable~\citep{Tulsiani09:CSP}, \Cref{lemma:IS} allow us to preclude a sparsity up to $\approx \sqrt{n}$. We will next show how to tighten that lower bound. First, it is worth noting that for CCE, one can improve the welfare bound of~\Cref{lemma:completeness} to $1 + \gamma T/k$~\citep{Anonymous24:Barriers}---that is, the welfare lower bound increases with the number of mixtures. As a result, a similar argument to~\Cref{lemma:IS} allows one to extract an independent set in $G$ with size $\Omega(k/T)$, as opposed to $\Omega(k/T^2)$. 

As we alluded to earlier, while the proof of~\Cref{lemma:IS} only makes use of external deviations (via~\Cref{lemma:ub-prob}), the key lies in appropriately leveraging internal deviations. Let us first provide some basic intuition. \Cref{lemma:ub-prob} makes use of the obvious inequality $\sum_{t=1}^T \vx_i^{(t)} \geq \max_{t \leq T } \vx_i^{(t)}$, which is only tight when at most a single product distribution of $\vmu$ is supported on $i$. While this case does not have any clear implications for CCE, it does so for CE: when Player $x$ gets a recommendation to play action $i$, Player $x$ can infer the underlying product distribution, and thereby deviate accordingly to the $\frac{k}{2} \mat{I}_n$ component in the off-diagonal of $(\mat{R}, \mat{C})$. (In the language of internal deviations, Player $x$ can deviate differently when at $i$.) In turn, this now provides a strong upper bound on $\vy_i^{(t)}$. In other words, when $\vx_i^{(t)}$ is large, we can sharpen our upper bound (compared to~\Cref{lemma:ub-prob}) on $\vy_i^{(t)}$. This is illustrated in~\Cref{fig:CEs} when $T = 8$ (leftmost image), $T = 4$ (center image), and $T = 2$ (rightmost) image. In particular, let us focus on the correlated distribution portrayed in the leftmost image. We think of each row representing a different product distribution, so that $\vx_i^{(i)} = 1$. Player $x$ can deviate very effectively as any recommendation betrays the underlying product distribution, thereby forcing Player $y$ to spread out its probability mass for each individual product distribution. As we will formalize, this argument enables providing a tighter bound on $\vx_i^{(t)} \vy_i^{(t)}$, instead of bounding each of those terms separately (\Cref{lemma:ref-ub}). A careful examination of the proof of~\Cref{lemma:IS}  reveals that such a bound suffices, as formalized in~\Cref{lemma:ref-IS}.

\begin{figure}[!ht]
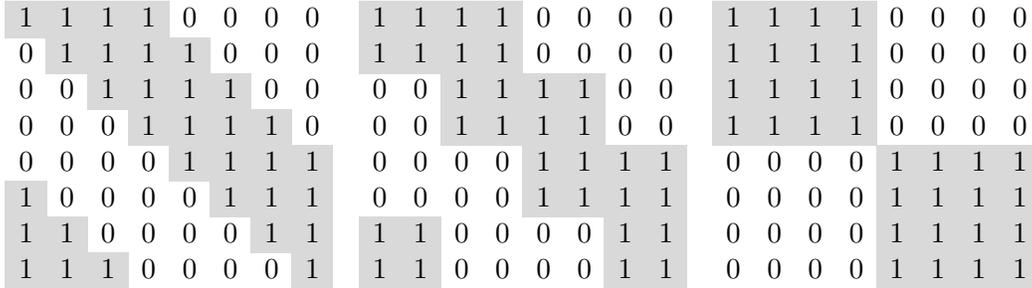

    \[
\begin{array}{ccc}
\begin{array}{cccccccc}
\cellcolor{lightgray}1 & \cellcolor{lightgray}1 & \cellcolor{lightgray}1 & \cellcolor{lightgray}1 & 0 & 0 & 0 & 0 \\
0 & \cellcolor{lightgray}1 & \cellcolor{lightgray}1 & \cellcolor{lightgray}1 & \cellcolor{lightgray}1 & 0 & 0 & 0 \\
0 & 0 & \cellcolor{lightgray}1 & \cellcolor{lightgray}1 & \cellcolor{lightgray}1 & \cellcolor{lightgray}1 & 0 & 0 \\
0 & 0 & 0 & \cellcolor{lightgray}1 & \cellcolor{lightgray}1 & \cellcolor{lightgray}1 & \cellcolor{lightgray}1 & 0 \\
0 & 0 & 0 & 0 & \cellcolor{lightgray}1 & \cellcolor{lightgray}1 & \cellcolor{lightgray}1 & \cellcolor{lightgray}1 \\
\cellcolor{lightgray}1 & 0 & 0 & 0 & 0 & \cellcolor{lightgray}1 & \cellcolor{lightgray}1 & \cellcolor{lightgray}1 \\
\cellcolor{lightgray}1 & \cellcolor{lightgray}1 & 0 & 0 & 0 & 0 & \cellcolor{lightgray}1 & \cellcolor{lightgray}1 \\
\cellcolor{lightgray}1 & \cellcolor{lightgray}1 & \cellcolor{lightgray}1 & 0 & 0 & 0 & 0 & \cellcolor{lightgray}1 \\
\end{array}
&
\begin{array}{cccccccc}
\cellcolor{lightgray}1 & \cellcolor{lightgray}1 & \cellcolor{lightgray}1 & \cellcolor{lightgray}1 & 0 & 0 & 0 & 0 \\
\cellcolor{lightgray}1 & \cellcolor{lightgray}1 & \cellcolor{lightgray}1 & \cellcolor{lightgray}1 & 0 & 0 & 0 & 0 \\
0 & 0 & \cellcolor{lightgray}1 & \cellcolor{lightgray}1 & \cellcolor{lightgray}1 & \cellcolor{lightgray}1 & 0 & 0 \\
0 & 0 & \cellcolor{lightgray}1 & \cellcolor{lightgray}1 & \cellcolor{lightgray}1 & \cellcolor{lightgray}1 & 0 & 0 \\
0 & 0 & 0 & 0 & \cellcolor{lightgray}1 & \cellcolor{lightgray}1 & \cellcolor{lightgray}1 & \cellcolor{lightgray}1 \\
0 & 0 & 0 & 0 & \cellcolor{lightgray}1 & \cellcolor{lightgray}1 & \cellcolor{lightgray}1 & \cellcolor{lightgray}1 \\
\cellcolor{lightgray}1 & \cellcolor{lightgray}1 & 0 & 0 & 0 & 0 & \cellcolor{lightgray}1 & \cellcolor{lightgray}1 \\
\cellcolor{lightgray}1 & \cellcolor{lightgray}1 & 0 & 0 & 0 & 0 & \cellcolor{lightgray}1 & \cellcolor{lightgray}1 \\
\end{array}
&
\begin{array}{cccccccc}
\cellcolor{lightgray}1 & \cellcolor{lightgray}1 & \cellcolor{lightgray}1 & \cellcolor{lightgray}1 & 0 & 0 & 0 & 0 \\
\cellcolor{lightgray}1 & \cellcolor{lightgray}1 & \cellcolor{lightgray}1 & \cellcolor{lightgray}1 & 0 & 0 & 0 & 0 \\
\cellcolor{lightgray}1 & \cellcolor{lightgray}1 & \cellcolor{lightgray}1 & \cellcolor{lightgray}1 & 0 & 0 & 0 & 0 \\
\cellcolor{lightgray}1 & \cellcolor{lightgray}1 & \cellcolor{lightgray}1 & \cellcolor{lightgray}1 & 0 & 0 & 0 & 0 \\
0 & 0 & 0 & 0 & \cellcolor{lightgray}1 & \cellcolor{lightgray}1 & \cellcolor{lightgray}1 & \cellcolor{lightgray}1 \\
0 & 0 & 0 & 0 & \cellcolor{lightgray}1 & \cellcolor{lightgray}1 & \cellcolor{lightgray}1 & \cellcolor{lightgray}1 \\
0 & 0 & 0 & 0 & \cellcolor{lightgray}1 & \cellcolor{lightgray}1 & \cellcolor{lightgray}1 & \cellcolor{lightgray}1 \\
0 & 0 & 0 & 0 & \cellcolor{lightgray}1 & \cellcolor{lightgray}1 & \cellcolor{lightgray}1 & \cellcolor{lightgray}1 \\
\end{array}
\end{array}
\]
    \caption{Different correlated distributions when $k = 4$ (up to the normalization constant $32$).}
    \label{fig:CEs}
\end{figure}

\begin{lemma}
    \label{lemma:ref-ub}
    Let $\vmu \defeq \frac{1}{T} \sum_{t=1}^T \vx^{(t)} \otimes \vy^{(t)}$ be an $\epsilon$-CE of $(\mat{R}, \mat{C})$, with $\epsilon \leq \frac{1}{k^2}$ and $\gamma \leq \frac{1}{4}$, supported only on $[n] \times [n]$. Then, for any $i \in [n]$ and $t \in [T]$, $\vx_i^{(t)} \vy_i^{(t)} \leq 4 T / k^2$.
\end{lemma}

\begin{proof}
    In the proof~\Cref{lemma:ub-prob}, we have already seen that for any $i \in [n]$,
    \begin{equation}
        \label{eq:devy}
        \sum_{t=1}^T \vx_i^{(t)} \leq T \frac{1 + \gamma + 2 \epsilon}{k}.
    \end{equation}
    Let us fix some $i \in [n]$. We consider the internal deviation of Player $x$ in which $i \mapsto i + n$ and $i' \mapsto i'$ for any $i' \neq i$. The utility under that deviation can be expressed as 
    \[
        \frac{k}{2} \frac{1}{T} \sum_{t=1}^T \vx_i^{(t)} \vy_i^{(t)} + \frac{1}{T} \sum_{t=1}^T \sum_{i' \neq i} \sum_{j=1}^n \vx_{i'}^{(t)} \mat{R}_{i', j} \vy_j^{(t)}.
    \]
    Given that $\vmu$ is assumed to be an $\epsilon$-CE, we have
    \begin{align*}
        \frac{k}{2} \frac{1}{T} \sum_{t=1}^T \vx_i^{(t)} \vy_i^{(t)} + \frac{1}{T} \sum_{t=1}^T \sum_{i' \neq i} \sum_{j=1}^n \vx_{i'}^{(t)} \mat{R}_{i', j} \vy_j^{(t)} \leq  \frac{1}{T} \sum_{t=1}^T \sum_{i' = 1}^n \sum_{j=1}^n \vx_{i'}^{(t)} \mat{R}_{i', j} \vy_j^{(t)} + \epsilon \\
        \iff \frac{k}{2} \frac{1}{T} \sum_{t=1}^T \vx_i^{(t)} \vy_i^{(t)} \leq \frac{1}{T} \sum_{t=1}^T \vx_i^{(t)} \sum_{j=1}^n \mat{R}_{i, j} \vy_j^{(t)} + \epsilon.
    \end{align*}
    Thus, for any $t \in [T]$,
    \begin{equation}
        \label{eq:eps-care}
        \vy_i^{(t)} \frac{ \vx_i^{(t)} }{ \sum_{t=1}^T \vx_i^{(t)} } \leq \frac{1 + \gamma }{k} + T \frac{2\epsilon}{k \vx_i^{(t)}}.
    \end{equation}
    If $\vx_i^{(t)} \leq 4 T/k^2$, then clearly $\vx_i^{(t)} \vy_i^{(t)} \leq 4 T/k^2$. We can thus treat the case where $\vx_i^{(t)} > 4 T/k^2$. Continuing from~\eqref{eq:eps-care},
    \begin{equation}
        \label{eq:final-devx}
        \vy_i^{(t)} \frac{ \vx_i^{(t)} }{ \sum_{t=1}^T \vx_i^{(t)} } \leq \frac{1 + \gamma }{k} + \frac{\epsilon k}{2} \leq \frac{1.5 + \gamma}{k},
    \end{equation}
    where we used the assumption that $\epsilon \leq \frac{1}{k^2}$. Combining~\eqref{eq:devy} and~\eqref{eq:final-devx}, we have
    \[
        \vx_i^{(t)} \vy_i^{(t)} \leq 4 \frac{T}{k^2},
    \]
    as desired.
\end{proof}
In contrast, the upper bound on $\vx_i^{(t)} \vy_i^{(t)}$ implied by~\Cref{lemma:ub-prob} scales as $\Theta(T^2/k^2)$. We will now show how to leverage~\Cref{lemma:ref-ub} to improve~\Cref{lemma:IS}.

\begin{lemma}
    \label{lemma:ref-IS}
    Suppose that there is a uniform $T$-sparse $\frac{1}{4k^2}$-CE in $(\mat{R}, \mat{C})$ supported only on $[n] \times [n]$ with welfare at least $1 + 3\gamma/(4 k)$. Then, if $\gamma = 1/(256 k T^3)$, there is an independent set in $G$ with size $k/(16 T)$.
\end{lemma}

\begin{proof}
    The argument follows the proof of~\Cref{lemma:IS} in conjunction with the refinement of~\Cref{lemma:ref-ub}. By~\Cref{lemma:ub-prob} and our choice of parameters, it follows that $\vx_i^{(t)}, \vy_j^{(t)} \leq 2 T / k$ for any $i, j \in [n]$. By selecting $S \defeq \{i \in [n] : \vx_i^{(t)}, \vy_i^{(t)} \geq 1/(16 k T) \}$, we have
    \begin{align*}
        \frac{3}{4k} &\leq \sum_{i \in S} \vx_i^{(t)} \vy_i^{(t)} + \sum_{i \notin S} \vx_i^{(t)} \vy_i^{(t)} \\
        &\leq |S| \frac{4T}{k^2} + \frac{1}{2k},
    \end{align*}
    where we used~\Cref{lemma:ref-ub}. This implies that $|S| \geq k/(16 T)$. Further, as in the proof of~\Cref{lemma:IS}, it follows that $S$ is an independent set in $G$.
\end{proof}

To finish the construction of an $\soshard$ game, we rely on the following SoS lower bound for the independent set problem due to~\citet{Tulsiani09:CSP}; although it was phrased as an integrality gap in terms of the objective value of the SDP, the decision version we state below can be readily extracted (see the lecture notes of~\citet{Barak16:Proofs}).

\begin{theorem}[\citep{Tulsiani09:CSP}]
    \label{theorem:Tulsiani-IS}
    For a large enough $n \in \N$ and constants $c_1, c_2 > 0$, there is a family of $n$-node graphs such that
    \begin{itemize}
    \item there is a degree-$2^{c_2 \sqrt{\log n \log \log n} }$ pseudo-expectation that satisfies the constraints of an independent set with size $k = k(n)$ on $G$, and
        \item every independent set in $G$ has size at most $k 2^{c_1 \sqrt{ \log n \log \log n }} / n $.
    \end{itemize}
\end{theorem}

We now return to the proof of completeness. To construct a pseudo-NE that satisfies the desired utility constraints starting from~\Cref{theorem:Tulsiani-IS}, we will use a standard approach in the context of SoS reductions~\citep{Tulsiani09:CSP} (see also~\citep[Fact 5.6]{Kothari18:Sum} and the lecture notes of~\citet{Barak16:Proofs} for further background). Namely, let $\Psi$ be the system corresponding to the independent set problem with size $k$, $\{ \vz_i^2 = \vz_i \}_{i \in [n]}$, $\{ \sum_{i=1}^n \vz_i = k \}$, and $\{ \vz_i \vz_j = 0 \}_{ \{i, j\} \in E }$, and let $\tilde{\nu}$ be a degree-$d$ pseudo-distribution that satisfies $\Psi$. In accordance with~\Cref{lemma:completeness}, we construct a pseudo-distribution $\tilde{\rho}$ on $(\vx, \vy)$ of $(\mat{R}, \mat{C})$ based on the mapping $\vx_i = \frac{1}{k} \vz_i $ for $i \leq n$ and $0$ otherwise, and $\vy_j = \frac{1}{k} \vz_j$ for $j \leq n$ and $0$ otherwise; we denote this mapping by $(\vx, \vy) = (\vx(\vz), \vy(\vz))$. This yields the associated pseudo-expectation
\begin{equation}
    \label{eq:map1}
    f \mapsto \tilE_{\tilde{\rho}} [ f(\vx, \vy) ] = \tilE_{\tilde{\nu}} [ f(\vx(\vz), \vy(\vz))].
\end{equation}
Now, let $\Psi'$ be the system corresponding to $1$-sparse CE in which both players obtain a utility of at least $\delta$. Following~\Cref{lemma:completeness}, it is a straightforward exercise to show that $\tilde{\rho}$ is an $\Omega(d)$-degree pseudo-distribution that satisfies the constraints of $\Psi'$; this is essentially the same property shown by~\citet{Kothari18:Sum} for Nash equilibria. Finally, based on $\tilde{\rho}$, we construct a pseudo-distribution $\tilmu$ on $(\vx^{(1)}, \dots, \vx^{(T)}, \vy^{(1)}, \dots, \vy^{(T)})$, such that
\begin{equation}
    \label{eq:map2}
    f \mapsto \tilE_{\tilmu} [ f(\vx^{(1)}, \dots, \vx^{(T)}, \vy^{(1)}, \dots, \vy^{(T)})] = \tilE_{\tilde{\rho}} [ f(\vx, \dots, \vx, \vy, \dots, \vy)].
\end{equation}
Then, it follows that $\tilmu$ is a degree-$\Omega(d)$ pseudo-distribution on $(\vx^{(1)}, \dots, \vx^{(T)}, \vy^{(1)}, \dots, \vy^{(T)})$ that satisfies the constraints of $T$-sparse CE together with the constraint that both players obtain a utility of at least $\delta$. We gather the transformations above in the following definition.

\begin{definition}
    \label{def:tilmu}
    Let $\tilde{\nu}$ be a pseudo-distribution that satisfies the constraints of the independent set problem with size $k$ on $G$. We define a pseudo-distribution $\tilde{\mu}$ on $(\vx^{(1)}, \dots, \vx^{(T)}, \vy^{(1)}, \dots, \vy^{(T)})$, with respect to game $(\mat{R}(G), \mat{C}(G))$ in~\eqref{eq:indep-setgame}, according to~\eqref{eq:map1} and~\eqref{eq:map2}.
\end{definition}

We are now ready to establish the existence of an $\soshard$ game. We start from a graph described in~\Cref{theorem:Tulsiani-IS}. When $T = n/(32 \cdot 2^{ c_1 \sqrt{\log n \log \log n}}) = n^{1 - o(1)}$ and $\gamma = 1/(256 k T^3)$, \Cref{lemma:bad-prob,lemma:ref-IS} imply that any uniform $T$-sparse $O(1/k^2)$-CE in the induced game~\eqref{eq:indep-setgame} must have welfare at most $2 \delta - \gamma/(4k)$. Further, $\tilmu$ described in~\Cref{def:tilmu} satisfies all the properties of~\Cref{def:soshard} with degree given by~\Cref{theorem:Tulsiani-IS}, leading to the following result. (Rescaling~\eqref{eq:indep-setgame} by $1/k$ only affects constants $c_2$ and $c_3$ below; doing so also rescales the value of $\delta$ in~\Cref{def:soshard}.)

\begin{theorem}[$\soshard$ game for high precision]
    \label{theorem:soshard-high}
    For some absolute constants $c_1, c_2, c_3$, there is a $(2^{c_1 \sqrt{\log n \log \log n} }, n^{1 - o(1)}, n^{-c_2}, n^{-c_3})$-$\soshard$ game. 
\end{theorem}

\subsection{Low precision}

Our approach in the low-precision regime is similar, but instead relies on SoS lower bounds for \emph{planted clique}---a central problem in average-case complexity. We should note that~\citet{Kothari18:Sum} had to rely on a different reduction~\citep{Braverman15:Approximating,Deligkas18:Inapproximability} since the result we use below (\Cref{theorem:sos-planted}) was not available at the time.

In particular, we will use a result by~\citet{Pang21:SOS} (see also~\citep{Meka15:Sum,Deshpande15:Improved,Hopkins16:Integrality,Barak19:Nearly} for some earlier results). It is assumed that an $n$-node graph $G$ is drawn from an Erd\H{o}s-R\'enyi family $G(n, \frac{1}{2})$, meaning that each edge is included independently with probability $\frac{1}{2}$. An important property is that in such graphs, there is not even a $2.1 \log n$ clique with high probability---that is, with probability at least $1 - n^{-c}$. Yet, \citet{Pang21:SOS} established the following result.

\begin{theorem}[\citep{Pang21:SOS}]
    \label{theorem:sos-planted}
    With high probability over $G(n, \frac{1}{2})$, there is a degree-$\Omega(\log n / \log \log n)$ pseudo-distribution satisfying the constraints of $k$-clique for any $k \leq n^{1/2 - O(1)}$. 
\end{theorem}

It is likely that the degree above can be improved to $\Omega(\log n)$~\citep{Pang21:SOS}, but to our knowledge that remains open. Now, the reduction we described earlier in~\Cref{sec:wel-highprec} based on~\eqref{eq:indep-setgame} can be readily applied here as well. In particular, although it was formulated based on the independent set problem, one can simply replace $\mat{1}_{n \times n} - \mat{A}$ by $\mat{A} - \mat{I}_n$ to translate the previous lemmas in terms of the clique problem. As a result, using~\Cref{theorem:sos-planted} with $k = \log^{O(1)} n$ we arrive at the following conclusion. (As before, rescaling the utilities by $1/k$ affects constants $c_2$ and $c_3$ below.)

\begin{theorem}[$\soshard$ game for low precision]
    \label{theorem:soshard-low}
    For some absolute constants $c_1, c_2, c_3$, there is a $(c_1 \log n / \log \log n, \log n, (\log n)^{-c_2}, (\log n)^{-c_3} )$-$\soshard$ game.
\end{theorem}

As an alternative to the reduction presented in~\Cref{sec:wel-highprec}, it is worth pointing out a result due to~\citet{Hazan11:How} (\emph{cf.}~\citep{Austrin13:Inapproximability}) which precludes finding a Nash equilibrium maximizing welfare even when $\epsilon = \Theta(1)$ (under the planted clique conjecture). It can be shown that their reduction can be extended to $T$-sparse CE when $T = O(1)$~\citep{Anonymous24:Barriers}.
\section{Hard games under enumeration}
\label{sec:enum}

In this section, we construct a family of $\enumhard$ games with respect to (coarse) correlated equilibria by relying on analogous lower bounds for Nash equilibria~\citep{Kothari18:Sum,Daskalakis09:Oblivious}. We begin by stating precisely the definition of $\enumhard$ games.

\begin{definition}
    \label{def:enumhard}
    A family of (two-player) games $\mathcal{F} = \{(\mat{R}^S, \mat{C}^S) \}_{S \in \mathcal{S}}$ with utilities in $[-1, 1]$ is $\epsilon$-$\enumhard$ with respect to the set of CCE if for any $S, S' \in \mathcal{S}$ with $S \neq S'$, if $\vmu$ is an $\epsilon$-CCE of $(\mat{R}^S, \mat{C}^S)$ and $\vmu'$ is an $\epsilon$-CCE of $(\mat{R}^{S'}, \mat{C}^{S'})$, then $\vmu \neq \vmu'$.
\end{definition}

The rationale behind this definition is the following. Consider a potentially randomized algorithm $\mathcal{A}$ that is able to output an $\epsilon$-CCE with some constant probability, but can only access the game via the verification oracle of~\Cref{def:VO}. If there is an underlying family of $\epsilon$-$\enumhard$ games per~\Cref{def:enumhard}, then each call to the verification oracle can only exclude a single game at a time---when it returns \reject. But $\mathcal{A}$ has to be able to identify the underlying game within the class of games $\mathcal{F}$ with constant probability in order to output an $\epsilon$-CCE. As a result, $\mathcal{A}$ must invoke the verification oracle a number of times proportional to $|\mathcal{F}|$ in order to be successful with constant probability. In other words, the key quantity of interest behind~\Cref{def:enumhard} is the cardinality of $\mathcal{F}$ as a function of $\epsilon$.

We begin with the regime where the approximation parameter $\epsilon$ is a constant (low precision), and we then proceed with the regime where $\epsilon = \poly(1/n)$ (high precision). 

\subsection{Low precision}

In the regime where the approximation parameter $\epsilon$ is a constant, we will adapt the lower bound of~\citet{Daskalakis09:Oblivious}, which in turn is based on the earlier construction by~\citet{Althofer94:Sparse}; the main motivation behind the lower bound of~\citet{Daskalakis09:Oblivious} was to prove that the quasipolynomial-time algorithm of~\citet{Lipton03:Playing} is essentially optimal for the class of ``oblivious algorithms.'' The main result pertaining to $\enumhard$ in the low-precision regime is summarized below; \Cref{item:default} is an auxiliary property not related to $\enumhard$ games, but will be useful later when combining games (\Cref{sec:stitching}).

\begin{theorem}[$\enumhard$ games for low precision]
    \label{theorem:enumhard-low}
    There is a family $\mathcal{F}$ of $n^{\Omega(\log n)}$ two-player $n \times n$ games $\{ (\mat{R}^S, \mat{C}^S) \}_{S \in \mathcal{S}}$ and a sufficiently small $\epsilon = \Theta(1)$ such that
    \begin{enumerate}
        \item $\mathcal{F}$ is $\epsilon$-$\enumhard$, and \label{item:enum-hard}
        \item for any $S \in \mathcal{S}$, there is a strategy $\vxdef \in \Delta^n$ such that $\langle \vxdef, \mat{R}^S \vy \rangle \geq \frac{1}{2}$ for any $\vy \in \Delta^n$, and a strategy $\vydef \in \Delta^n$ such that $\langle \vx, \mat{C}^S \vydef \rangle \geq \frac{1}{2}$ for any $\vx \in \Delta^n$. \label{item:default}
    \end{enumerate}
\end{theorem}

In this context, we let $\ell$ be an even integer and define $n \defeq \binom{\ell}{\ell/2}$. For each subset $S \subseteq [n]$ with cardinality $\ell$, we construct an $n \times n$ game $(\mat{R}^S, \mat{C}^S)$ indexed by $S$. Before we describe the game, we note that the number of subsets of $S$ with cardinality $\ell/2$ is exactly $n$; the strategy of Player $y$ can be thought of as indexing each of those subsets. Under column $j$, corresponding to $S_j \subseteq S$ with $|S_j| = \ell/2$, the utilities of $(\mat{R}^S, \mat{C}^S)$ are defined as follows.

\begin{itemize}
    \item For each $i \notin S$, Player $x$ gets utility $-1$ and Player $y$ gets utility $1$; 
    \item for each $i \in S \cap S_j$, Player $x$ gets utility $1$ and Player $y$ gets utility $0$; and 
    \item otherwise, Player $x$ gets utility $0$ and Player $y$ gets utility $1$.
\end{itemize}

To clarify this construction, in~\Cref{figure:dask} we give an example of $\mat{R}^S$ when $n = 6$, $\ell = 4$, and $S = \{1, 2, 3, 4 \}$~\citep[Figure 1]{Daskalakis09:Oblivious}.

\begin{figure}[!ht]
    \small
    \centering
    \begin{tabular}{ccccccc}
        & \{1, 2\} & \{1, 3\} & \{1, 4\} & \{2, 3\} & \{2, 4\} & \{3, 4\} \\
        1 & 1 & 1 & 1 & 0 & 0 & 0 \\
        2 & 1 & 0 & 0 & 1 & 1 & 0 \\
        3 & 0 & 1 & 0 & 1 & 0 & 1 \\
        4 & 0 & 0 & 1 & 0 & 1 & 1 \\
        \cellcolor{lightgray}5 & \cellcolor{lightgray}-1 & \cellcolor{lightgray}-1 & \cellcolor{lightgray}-1 & \cellcolor{lightgray}-1 & \cellcolor{lightgray}-1 & \cellcolor{lightgray}-1 \\
        \cellcolor{lightgray}6 & \cellcolor{lightgray}-1 & \cellcolor{lightgray}-1 & \cellcolor{lightgray}-1 & \cellcolor{lightgray}-1 & \cellcolor{lightgray}-1 & \cellcolor{lightgray}-1 \\
    \end{tabular}
    \caption{An example of $\mat{R}^S$ when $n = 6$, $\ell = 4$, and $S = \{1, 2, 3, 4\}$. Entries highlighted in gray correspond to strictly dominated actions for Player $x$.}
    \label{figure:dask}
\end{figure}

In particular, we see that actions for Player $x$ outside of $S$ are strictly dominated by ones in $S$. In an approximate CCE, it can be shown that Player $x$ will not assign a significant probability mass to such actions. Conditioning on $i \in S$, we see that the induced game is $1$-sum. This is quite crucial for our purposes as CCE and Nash equilibria behave similarly in such games, in the following precise sense.

\begin{fact}
    \label{fact:zerosum}
    Suppose that a two-player game $(\mat{R}, \mat{C})$ is constant-sum. If $\vmu = \frac{1}{T} \sum_{t=1}^T \vx^{(t)} \otimes \vy^{(t)}$ is an $\epsilon$-CCE of $(\mat{R}, \mat{C})$, then the pair $(\Bar{\vx}^{(T)}, \bar{\vy}^{(T)})$ is a $2 \epsilon$-Nash equilibrium of $(\mat{R}, \mat{C})$, where $\Bar{\vx}^{(T)} = \frac{1}{T} \sum_{t=1}^T \vx^{(t)}$ and $\Bar{\vy}^{(T)} = \frac{1}{T} \sum_{t=1}^T \vy^{(t)}$.
\end{fact}

This fact follows readily from the definitions, as we explain below.

\begin{proof}[Proof of~\Cref{fact:zerosum}]
    We can assume that $(\mat{R}, \mat{C})$ is zero-sum since one can always shift the utilities without altering the set of equilibria. Then, there is $(\vx, \vy) \in \Delta^n \times \Delta^n$ such that
    \begin{align*}
        0 = \E_{(i,j) \sim \vmu} [ \mat{R}_{i, j} + \mat{C}_{i, j}] &\geq \langle \vx, \mat{R} \bar{\vy}^{(T)} \rangle + \langle \bar{\vx}^{(T)}, \mat{C} \vy \rangle - 2 \epsilon \\
        &\geq \max_{\vx \in \Delta^n} \langle \vx, \mat{R} \bar{\vy}^{(T)} \rangle - \langle \bar{\vx}^{(T)}, \mat{R} \bar{\vy}^{(T)} \rangle + \max_{\vy \in \Delta^n} \langle \bar{\vx}^{(T)}, \mat{C} \vy \rangle - \langle \bar{\vx}^{(T)}, \mat{C} \bar{\vy}^{(T)} \rangle - 2 \epsilon,
    \end{align*}
    where we used the $\epsilon$-CCE constraint for both players in order to deviate to $(\vx, \vy)$. The last inequality implies that $(\bar{\vx}^{(T)}, \bar{\vy}^{(T)})$ is a $2\epsilon$-Nash equilibrium of $(\mat{R}, \mat{C})$, as claimed.
\end{proof}

We can thus follow a similar approach to~\citet{Daskalakis09:Oblivious} in order to show that the marginal of any CCE of $(\mat{R}^S, \mat{C}^S)$ for Player $x$ must be close to the uniform distribution over $S$. Below, we use the notation $\vu(S) \in \Delta^n$ for the uniform distribution over $S \subseteq [n]$; the domain of $\vu(S)$ will be clear from the context, and it is thus left unspecified.

\begin{lemma}
    \label{lemma:near-uni}
    If $\vmu = \frac{1}{T} \sum_{t=1}^T \vx^{(t)} \otimes \vy^{(t)}$ is an $\epsilon$-CCE of $(\mat{R}^S, \mat{C}^S)$ for a sufficiently small $\epsilon = \Theta(1)$, then
    \[
        \left\| \frac{1}{T} \sum_{t=1}^T \vx^{(t)} - \vu(S) \right\|_1 \leq O(\epsilon).
    \]
\end{lemma}

We will rely on the following characterization, which can be extracted from~\citep[Lemma 2]{Daskalakis09:Oblivious}. We caution that~\citet{Daskalakis09:Oblivious} work with the stronger notion of \emph{well-supported} Nash equilibria, but the claim below follows readily from their argument.

\begin{lemma}[\citep{Daskalakis09:Oblivious}]
    \label{lemma:WNE}
    If $(\vx, \vy)$ is an $\epsilon$-NE of $(\mat{R}_{S, :}^S, \mat{C}_{S, :}^S)$ for $\epsilon < 1$, then $\| \vx - \vu(S) \|_1 \leq 8 \epsilon$.
\end{lemma}

Above, $\mat{R}^S_{S, :}$ denotes the submatrix of $\mat{R}^S$ containing all columns of $\mat{R}^S$ and the rows corresponding to $S$. We are now ready to prove~\Cref{lemma:near-uni}.

\begin{proof}[Proof of~\Cref{lemma:near-uni}]
    First, by construction of $(\mat{R}^S, \mat{C}^S)$, we have
    \begin{align}
        \frac{1}{T} \sum_{t=1}^T \langle \vx^{(t)}, \mat{C}_S \vy^{(t)} \rangle &= \frac{1}{T} \sum_{t=1}^T \sum_{i \notin S} \vx^{(t)}_i + \frac{1}{T} \sum_{t=1}^T \langle \vx^{(t)}_S, \mat{C}^{S}_{S, :} \vy^{(t)} \rangle \notag \\ &\geq \frac{1}{T} \sum_{t=1}^T \sum_{i \notin S} \vx_i^{(t)} + \frac{1}{2 T} \sum_{t=1}^T \sum_{i \in S} \vx_i^{(t)} - \epsilon,\label{align:col-bound}
    \end{align}
    where the inequality above follows by switching to the uniform strategy over $[n]$ for Player $y$ in the definition of CCE. Thus, since $(\mat{R}^{S}_{S, :}, \mat{C}^{S}_{S, :})$ is a $1$-sum game, \eqref{align:col-bound} can be expressed as
    \[
        \frac{1}{T} \sum_{t=1}^T \langle \vx^{(t)}_S, \mat{R}^{S}_{S, :} \vy^{(t)} \rangle = 1 - \frac{1}{T} \sum_{t=1}^T \langle \vx^{(t)}_S, \mat{C}^{S}_{S, :} \vy^{(t)} \rangle \leq 1 - \frac{1}{2T} \sum_{t=1}^T \sum_{i \in S} \vx_i^{(t)} + \epsilon.
    \]
    So,
    \begin{align*}
        \frac{1}{T} \sum_{t=1}^T \langle \vx^{(t)}, \mat{R}^S \vy^{(t)} \rangle &= - \frac{1}{T} \sum_{t=1}^T \sum_{i \notin S} \vx_i^{(t)}  + \frac{1}{T} \sum_{t=1}^T \langle \vx^{(t)}_S, \mat{R}^{S}_{S, :} \vy^{(t)} \rangle \\
        &\leq - \frac{1}{2T} \sum_{t=1}^T \sum_{i \notin S} \vx_i^{(t)} + \frac{1}{2} + \epsilon.
    \end{align*}
    Moreover, we observe that $\langle \vu(S), \mat{R}^S \vy \rangle \geq \frac{1}{2}$ for any $\vy \in \Delta^n$, implying that $\frac{1}{T} \sum_{t=1}^T \langle \vx^{(t)}, \mat{R}^S \vy^{(t)} \rangle \geq \frac{1}{2} - \epsilon$ (via the CCE constraint with respect to Player $x$). We conclude that 
    \[
        \frac{1}{T} \sum_{t=1}^T \sum_{i \notin S} \vx_i^{(t)} \leq 4 \epsilon.
    \]
    We will now argue that $\frac{1}{T} \sum_{t=1}^T \hvx_S^{(t)} \otimes \vy^{(t)}$ is a $17 \epsilon$-CCE of $(\mat{R}^S_{S, :}, \mat{C}^S_{S, :})$, where $\hvx^{(t)} \in \Delta^n$ is such that $\hvx_S^{(t)} = \frac{1}{\| \vx^{(t)}_S \|_1} \vx^{(t)}_S$ for all $t \in [T]$. Indeed, we have
    \begin{align*}
        \left\| \frac{1}{T} \sum_{t=1}^T \vx^{(t)} - \frac{1}{T} \sum_{t=1}^T \hvx^{(t)} \right\|_1  &\leq 4 \epsilon + \frac{1}{T} \left\| \sum_{t=1}^T \sum_{i \in S} \left( \frac{1}{\| \vx_S^{(t)} \|_1} \vx_i^{(t)} - \vx_i^{(t)} \right) \right\|_1 \leq 8 \epsilon \\
        &\leq 4 \epsilon + \frac{1}{T} \sum_{t=1}^T \left( \frac{1}{\| \vx_S^{(t)} \|_1} - 1 \right) \sum_{i \in S} \vx_i^{(t)} \\
        &= 4 \epsilon + \frac{1}{T} \sum_{t=1}^T \sum_{i \notin S} \vx_i^{(t)} \\
        &= 8 \epsilon,
    \end{align*}
    and the claim follows since the entries of $\mat{R}^S$ and $\mat{C}^S$ are in $[-1, 1]$. Now, given that the game $(\mat{R}^S_{S, :}, \mat{C}^S_{S, :})$ is $1$-sum, \Cref{fact:zerosum} implies that $( \frac{1}{T} \sum_{t=1}^T \hvx_S^{(t)}, \frac{1}{T} \sum_{t=1}^T \vy^{(t)} )$ is a $17 \epsilon$-Nash equilibrium of $(\mat{R}^S_{S, :}, \mat{C}^S_{S, :})$. Together with~\Cref{lemma:WNE}, we have
    \[
        \left\| \frac{1}{T} \sum_{t=1}^T \vx^{(t)} - \vu(S) \right\|_1 \leq 8 \epsilon + \left\| \frac{1}{T} \sum_{t=1}^T \hvx_S^{(t)} - \vu(S) \right\|_1 = O(\epsilon).
    \]
\end{proof}

In conjunction with the previous lemma, the next observation will give rise to a large family of $\enumhard$ games.

\begin{lemma}[\citep{Daskalakis09:Oblivious}]
    \label{lemma:game-collection}
    For $\mathcal{S} \defeq \{S \subseteq [n], |S| = \ell \}$, it holds that $|\mathcal{S}| = \Omega(n^{ 0.8 \log_2 n })$. Further, there is a sufficiently small $\epsilon = \Theta(1)$ and a subset $\mathcal{S}' \subseteq \mathcal{S}$ of size $\Omega( n^{(0.8 - 2 \epsilon) \log_2 n } )$ such that for any $S, S' \in \mathcal{S}'$, $\|\vu(S) - \vu(S') \|_1 \geq \epsilon$.
\end{lemma}

We are now ready to prove~\Cref{theorem:enumhard-low}.

\begin{proof}[Proof of~\Cref{theorem:enumhard-low}]
    We let $\mathcal{S}'$ be as in~\Cref{lemma:game-collection}, and we define $\mathcal{F} \defeq \{ (\mat{R}^S, \mat{C}^S) \}_{S \in \mathcal{S'}}$. \Cref{item:default} follows immediately by taking $\vxdef = \vu(S)$ and $\vydef = \vu$ (where $\vu$ denotes the uniform distribution over $[n]$). We turn to~\Cref{item:enum-hard}. The claim on the cardinality of $\mathcal{F}$ follows from~\Cref{lemma:game-collection}. Now, for the sake of contradiction, suppose that there are two games $(\mat{R}^S, \mat{C}^S)$ and $(\mat{R}^{S'}, \mat{C}^{S'})$, with $S, S' \in \mathcal{S}'$, so that an $\epsilon$-CCE of the former coincides with some $\epsilon$-CCE of the latter. Then, their marginals for Player $x$ would also be the same, and~\Cref{lemma:near-uni} together with the triangle inequality imply that $\|\vu(S) - \vu(S') \|_1 = O(\epsilon)$. When $\epsilon = \Theta(1)$ is sufficiently small, this leads to a contradiction in light of~\Cref{lemma:game-collection}.
\end{proof}

\subsection{High precision}

In this regime, the starting point is the $m$-dimensional generalization of the matching pennies game, namely $(\mat{I}_m, - \mat{I}_{m})$. This game has the following property.

\begin{lemma}
    \label{lemma:gen-match}
    If $\vmu = \frac{1}{T} \sum_{t=1}^T \vx^{(t)} \otimes \vy^{(t)} \in \Delta([m] \times [m]) $ is an $\epsilon$-CCE of $(\mat{I}_m, -\mat{I}_m)$, then
    \[
        \left\| \frac{1}{T} \sum_{t=1}^T \vx^{(t)} - \vu \right\|_1 \leq 2 m^2 \epsilon.
    \]
\end{lemma}

\begin{proof}
We will first show that in an $\epsilon$-Nash equilibrium $(\vx, \vy)$ of that game, it must be the case that $\| \vx - \vu \|_1 \leq 2 m^2 \epsilon$. Indeed, it is not hard to see that the value (that is, the utility of Player $x$ in equilibrium) of $(\mat{I}_m, - \mat{I}_m)$ is $1/m$. For the sake of contradiction, suppose that there is some coordinate $\vx_i$ such that $\vx_i > 1/m + 2 \epsilon m$. Then, Player $y$ could play the uniform strategy over $[m] \setminus \{i\}$, which secures a cost of $ (1 - \vx_i) / (m - 1) < 1/m - 2 \epsilon m / (m-1) < 1/m - 2 \epsilon$ for that player. Further, if $\vx_i < 1/m - 2 \epsilon m$, Player $y$ can secure a cost of $1/m - 2 \epsilon m < 1/m - 2 \epsilon$ by playing $i$. So, under $(\vx, \vy)$, the utility of Player $x$ is less than $1/m - \epsilon$, which is a contradiction since $(\vx, \vy)$ is an $\epsilon$-Nash equilibrium and that player can also secure $1/m$ by switching to the uniform strategy over $[m]$. This argument establishes that $\|\vx - \vu \|_\infty \leq 2 \epsilon m$, in turn implying that $\| \vx - \vu \|_1 \leq 2 \epsilon m^2$. Finally, since $(\mat{I}_m, - \mat{I}_{m})$ is a zero-sum game, the claim follows from~\Cref{fact:zerosum}.
\end{proof}

In what follows, we will shift the utilities of that game so that it becomes $1$-sum and its value is $1/2$. This can be achieved by adding $1/2 - 1/m$ to the utilities of Player $x$ and $1/2 + 1/m$ to the utilities of Player $y$, while preserving the set of equilibria (and the corresponding approximation factors). Now, to construct $\enumhard$ games in this regime, we proceed as follows. For a set $S \subseteq [n]$ of cardinality $m$, we construct an $n \times n$ game $(\mat{R}^S, \mat{C}^S)$ such that
\begin{itemize}
    \item if $i \notin S$ and $j \notin S$ both players receive utility $-1$;
    \item if $i \in S$ and $j \notin S$ Player $x$ receives utility $1/2$ while Player $y$ receives utility $-1$;
    \item if $i \notin S$ and $j \in S$ Player $x$ receives utility $-1$ while Player $y$ receives utility $1/2$; and
    \item if $i \in S$ and $j \in S$ the utilities are defined per the generalized matching pennies game (under the shift in the utilities described above).
\end{itemize}

\begin{lemma}
    \label{lemma:enum-smalleps}
    If $\vmu = \frac{1}{T} \sum_{t=1}^T \vx^{(t)} \otimes \vy^{(t)} \in \Delta([n] \times [n])$ is an $\epsilon$-CCE of $(\mat{R}^S, \mat{C}^S)$, then
    \[
        \left\| \frac{1}{T} \sum_{t=1}^T \vx^{(t)} - \vu(S) \right\|_1 \leq O(n^2 \epsilon).
    \]
\end{lemma}

\begin{proof}
    By construction of $(\mat{R}^S, \mat{C}^S)$, we have
    \begin{align}
        \frac{1}{T} \sum_{t=1}^T \langle \vx^{(t)}, \mat{R}^S \vy^{(t)} \rangle &= - \frac{1}{T} \sum_{t=1}^T \sum_{i \notin S} \vx_i^{(t)} + \frac{1}{T} \sum_{t=1}^T \langle \vx_S^{(t)}, \mat{R}^S_{S, S} \vy_S^{(t)} \rangle + \frac{1}{2 T} \sum_{t=1}^T \sum_{i \in S} \vx_i^{(t)} \sum_{j \notin S} \vy_j^{(t)} \notag \\
        &\geq \frac{1}{T} \sum_{t=1}^T \langle \vu(S), \mat{R}^S \vy^{(t)} \rangle - \epsilon \notag \\
        &\geq \frac{1}{2} - \epsilon,\label{align:RSbound}
    \end{align}
    where the inequality follows from the definition of $\epsilon$-CCE, and we also used the fact that strategy $\vu(S)$ secures a utility of at least $\frac{1}{2}$ for Player $x$. Similarly,
    \begin{align}
        \frac{1}{T} \sum_{t=1}^T \langle \vx^{(t)}, \mat{C}^S \vy^{(t)} \rangle &= - \frac{1}{T} \sum_{t=1}^T \sum_{j \notin S} \vy_j^{(t)} + \frac{1}{T} \sum_{t=1}^T \langle \vx_S^{(t)}, \mat{C}^S_{S, S} \vy_S^{(t)} \rangle + \frac{1}{2T} \sum_{t=1}^T \sum_{i \notin S} \vx_i^{(t)} \sum_{j \in S} \vy_j^{(t)} \notag \\
        &\geq \frac{1}{T} \sum_{t=1}^T \langle \vx^{(t)}, \mat{C}^S \vu(S) \rangle - \epsilon \notag \\
        &\geq \frac{1}{2} - \epsilon. \label{align:CSbound}
    \end{align}
    Combining~\eqref{align:RSbound} and~\eqref{align:CSbound}, and using the fact that game $(\mat{R}^S_{S, S}, \mat{C}^S_{S, S})$ is $1$-sum, we have
    \[
        - \frac{1}{T} \sum_{t=1}^T \sum_{i \notin S} \vx_i^{(t)} - \frac{1}{T} \sum_{t=1}^T \sum_{j \notin S} \vy_j^{(t)} + \frac{1}{2 T} \sum_{t=1}^T \sum_{i \in S} \vx_i^{(t)} \sum_{j \notin S} \vy_j^{(t)} + \frac{1}{2 T} \sum_{t=1}^T \sum_{i \notin S} \vx_i^{(t)} \sum_{j \in S} \vy_j^{(t)} \geq - 2 \epsilon.
    \]
    Thus,
    \[
        \frac{1}{T} \sum_{t=1}^T \sum_{i \notin S} \vx_i^{(t)} + \frac{1}{T} \sum_{t=1}^T \sum_{j \notin S} \vy_j^{(t)} \leq 4 \epsilon.
    \]
    In turn, this implies that
    \[
        \left\| \frac{1}{T} \sum_{t=1}^T \vx^{(t)} - \frac{1}{T} \sum_{t=1}^T \hvx^{(t)} \right\|_1 \leq 8 \epsilon
    \]
    and 
    \[
        \left\| \frac{1}{T} \sum_{t=1}^T \vy^{(t)} - \frac{1}{T} \sum_{t=1}^T \hvy^{(t)} \right\|_1 \leq 8 \epsilon.
    \]
    Since the utilities of $(\mat{R}^S, \mat{C}^S)$ are in $[-1, 1]$, it follows that $\frac{1}{T} \sum_{t=1}^T \hvx_S^{(t)} \otimes \hvy_S^{(t)}$ is a $33 \epsilon$-CCE of $(\mat{R}^S_{S, S}, \mat{C}^S_{S, S})$, where $\hvx^{(t)}, \hvy^{(t)} \in \Delta^n$ are such that $\hvx_S^{(t)} = \frac{1}{\| \vx^{(t)}_S \|_1} \vx^{(t)}_S$ and $\hvy_S^{(t)} = \frac{1}{\| \vy^{(t)}_S \|_1} \vy^{(t)}_S$ for all $t \in [T]$. Together with~\Cref{lemma:gen-match} this gives the claim.
\end{proof}

Armed with this lemma, we consider the family $\mathcal{F} = \{ (\mat{R}^S, \mat{C}^S) \}_{S \subseteq [n], S \neq \emptyset}$. The second desideratum of~\Cref{def:enumhard} (\Cref{item:default}) is easily seen to be satisfied. For the first one, we note that $\|\vu(S) - \vu(S') \|_1 \geq \frac{1}{n}$ when $S \neq S'$. Therefore, taking $\epsilon$ in~\Cref{lemma:enum-smalleps} to be sufficiently small completes the construction in the high-precision regime as well.

\begin{theorem}[$\enumhard$ games for high precision]
    \label{theorem:enumhard-high}
    There is a family $\mathcal{F}$ of $2^n - 1$ two-player $n \times n$ games $\{ (\mat{R}^S, \mat{C}^S) \}_{S \in \mathcal{S}}$ and a sufficiently small $\epsilon = \Theta(1/n^3)$ such that 
    \begin{enumerate}
        \item $\mathcal{F}$ is $\epsilon$-$\enumhard$, and
        \item for any $S \in \mathcal{S}$, there is a strategy $\vxdef \in \Delta^n$ such that $\langle \vxdef, \mat{R}^S \vy \rangle \geq \frac{1}{2}$ for any $\vy \in \Delta^n$, and a strategy $\vydef \in \Delta^n$ such that $\langle \vx, \mat{C}^S \vydef \rangle \geq \frac{1}{2}$ for any $\vx \in \Delta^n$.
    \end{enumerate}
\end{theorem}
\section{Stitching games together under correlated equilibria}
\label{sec:stitching}

Given an $\soshard$ game $(\mat{R}, \mat{C})$ (\Cref{def:soshard}) and a game $(\mat{R}^S, \mat{C}^S)$ from a family of $\enumhard$ games (\Cref{def:enumhard}), both with utilities in $[-1, 1]$, this section shows how to combine the two in a single game with two key properties, to be defined in the sequel (\Cref{lemma:stitch-completeness,lemma:stitch-soundness}). The construction adapts the approach of~\citet{Kothari18:Sum}, and relies on the following $(2n) \times (2n)$ bimatrix game
\begin{equation}
    \label{eq:game-stitch}
    \mat{R}' =
    \begin{pmatrix}
        \mat{R} & - k \mat{1}_{n \times n} \\
        \delta \mat{1}_{n \times n} & \mat{R}^{S}
    \end{pmatrix} \quad
    \text{and} \quad
    \mat{C}' =
    \begin{pmatrix}
        \mat{C} & \delta \mat{1}_{n \times n} \\
        -k \mat{1}_{n \times n} & \mat{C}^S
    \end{pmatrix}.
\end{equation}
Here, we recall that $\mat{1}_{n \times n}$ is the all-ones $n \times n$ matrix, $\delta \in (0, \frac{1}{2})$ is associated with~\Cref{def:soshard}, and $k$ is a sufficiently large parameter. For the special case of Nash equilibria, it suffices to take $k$ to be an absolute constant, but in our more general setting we will take $k$ to be a function of other parameters as well. To translate our results in terms of additive approximations in a meaningful way, we will eventually have to normalize the range of $(\mat{R}', \mat{C}')$ to be in $[-1, 1]$.

The first key property of~\eqref{eq:game-stitch} is that a pseudo-CE with high welfare in $(\mat{R}, \mat{C})$ (as prescribed by~\Cref{def:soshard}) induces a pseudo-CE in $(\mat{R}', \mat{C}')$ which is the same invariantly from $(\mat{R}^S, \mat{C}^S)$---provided that the latter indeed comes from a family of $\enumhard$ games.

\begin{lemma}[Completeness]
    \label{lemma:stitch-completeness}
    There is a degree-$\Omega(d)$, uniform $T$-sparse pseudo-CE in $(\mat{R}', \mat{C}')$ for any $(\mat{R}^S, \mat{C}^S)$.
\end{lemma}

Indeed, starting from the pseudo-distribution $\tilmu$ on $\R^n \times \dots \times \R^n \ni (\hvx^{(1)}, \dots, \hvx^{(T)}, \hvy^{(1)}, \dots, \hvy^{(T)})$ described in~\Cref{def:tilmu}, we can extend it to a pseudo-distribution $\tilde{\nu}$ on $\R^{2n} \times \dots \times \R^{2n} \ni (\vx^{(1)}, \dots, \vx^{(T)}, \vy^{(1)}, \dots, \vy^{(T)})$ such that the associated pseudo-expectation reads
\[
    f \mapsto \tilE_{\tilmu} \left[f\left((\hvx^{(1)}, \vec{0}_n), \dots, (\hvx^{(T)}, \vec{0}_n), (\hvy^{(1)}, \vec{0}_n), \dots, (\hvy^{(T)}, \vec{0}_n)\right)\right].
\]
It is a straightforward exercise to show that $\tilde{\nu}$ defined above is a degree-$\Omega(d)$ pseudo-CE in $(\mat{R}', \mat{C}')$ no matter the choice of $(\mat{R}^S, \mat{C}^S)$; by the property of $\tilmu$ in~\eqref{eq:map2}, this proof of completeness is essentially the same as for Nash equilibria~\citep{Kothari18:Sum}.

We next concentrate on the soundness of the reduction. Namely, we will prove that approximate uniform $T$-sparse CE in $(\mat{R}', \mat{C}')$ are basically approximate uniform $T$-sparse CE of $(\mat{R}^S, \mat{C}^S)$ with only a small degradation in the approximation quality. We will split the argument into a series of claims.

We first observe that (C)CE in $(\mat{R}', \mat{C}')$ have a particular structure by virtue of the $-k \mat{1}_{n \times n}$ term in the off-diagonal of $\mat{R}' + \mat{C}'$: each product distribution comprising that correlated distribution is either supported on the component corresponding to $(\mat{R}, \mat{C})$ or on $(\mat{R}^S, \mat{C}^S)$, modulo some probability mass that becomes negligible when $k \gg 1$. Below, we define $\sumx^{(t)} \defeq \sum_{i=1}^n \vx_i^{(t)}$ and $\sumy^{(t)} \defeq \sum_{j=1}^n \vy_j^{(t)}$.

\begin{lemma}
    \label{claim:off-diagonal}
    Let $\vmu = \frac{1}{T} \sum_{t=1}^T \vx^{(t)} \otimes \vy^{(t)}$ be a $\frac{1}{2}$-CCE of $(\mat{R}', \mat{C}')$ with $k \geq 12 T$. Then, for any $t \in [T]$,
    \begin{itemize}
        \item either $\sumx^{(t)} \leq 6 T/k$ and $\sumy^{(t)} \leq 6T/k$,
        \item or $\sumx^{(t)} \geq 1 - 6T/k$ and $\sumy^{(t)} \geq 1 - 6 T/k$.
    \end{itemize}
\end{lemma}

\begin{proof}
    First, since $\vmu$ is a $\frac{1}{2}$-CCE of~\eqref{eq:game-stitch}, we have
    \begin{align}
        \sw(\vmu) = \frac{1}{T} \sum_{t=1}^T \langle \vx^{(t)}, (\mat{R}' + \mat{C}') \vy^{(t)} \rangle &\geq \frac{1}{T} \sum_{t=1}^T \langle (\vec{0}, \vxdef), \mat{R}' \vy^{(t)} \rangle + \frac{1}{T} \sum_{t=1}^T \langle \vx^{(t)}, \mat{C}' (\vec{0}, \vydef) \rangle - 1 \notag \\
        &\geq 2 \delta - 1,\label{align:sw-lb}
    \end{align}
    where we used the property of $\vxdef$ and $\vydef$ (per~\Cref{theorem:enumhard-high,theorem:enumhard-low}) and the fact that $\delta < \frac{1}{2}$. Further, having assumed that all entries in $(\mat{R}, \mat{C})$ and $(\mat{R}^S, \mat{C}^S)$ are in $[-1, 1]$, it follows that for any $t \in [T]$,
    \begin{equation}
        \label{eq:sw-ub}
        \sw(\vx^{(t)} \otimes \vy^{(t)}) \leq 2 - k \sumx^{(t)} (1 - \sumy^{(t)}) -k (1 - \sumx^{(t)}) \sumy^{(t)}.
    \end{equation}
    Combining~\eqref{align:sw-lb} and~\eqref{eq:sw-ub},
    \begin{align*}
        \frac{1}{T} \sum_{t=1}^T \left( \sumx^{(t)} (1 - \sumy^{(t)}) + (1 - \sumx^{(t)}) \sumy^{(t)} \right) \leq \frac{3}{k}.
    \end{align*}
    Thus, it follows that for any $t \in [T]$,
    \begin{equation}
        \label{eq:compl-bound}
        \sumx^{(t)} (1 - \sumy^{(t)}) + (1 - \sumx^{(t)}) \sumy^{(t)} \leq \frac{3T}{k}.
    \end{equation}
    For the sake of contradiction, let us suppose that $\sumx^{(t)} > \frac{3T}{k}$ and $\sumx^{(t)} < 1 - \frac{3T}{k}$. Then, $\sumx^{(t)} (1 - \sumy^{(t)}) + (1 - \sumx^{(t)}) \sumy^{(t)} > \frac{3T}{k} ( 1 - \sumy^{(t)} + \sumy^{(t)})$, contradicting~\eqref{eq:compl-bound}. Suppose now that $\sumx^{(t)} \leq \frac{3T}{k}$. We will show that the previous assertion implies that $\sumy^{(t)} \leq \frac{3T}{k}$. Indeed, by~\eqref{eq:compl-bound}, it follows that $ \frac{1}{2} \sumy^{(t)} \leq (1 - \sumx^{(t)}) \sumy^{(t)} \leq \frac{3T}{k}$, where we used the fact that $k \geq 6T$, which in turn implies that $\sumy^{(t)} \leq \frac{6T}{k}$. Finally, the case $\sumx^{(t)} \geq 1 - \frac{3T}{k}$ can be treated similarly, concluding the proof.
\end{proof}

Our next goal is to show that, in fact, the second case in the claim above can never arise. To do so, we first show that, in the contrary case, a CE in~\eqref{eq:game-stitch} gives rise to a CE in $(\mat{R}, \mat{C})$. It is worth noting that the argument below crucially hinges on $\Phi$ containing non-external deviations. The basic reason is that an external deviation in $(\mat{R}', \mat{C}')$ to some row of Player $x$ in $\mat{R}$ will incur a term proportional to $-k$ (assuming that $\setT \neq [T]$; see~\Cref{claim:induced-CE}); it thus follows that such deviations are non-profitable, and so it is unclear how to argue about deviations in $(\mat{R}, \mat{C})$ starting from a guarantee in $(\mat{R}', \mat{C}')$. One solution to this issue is to employ non-external deviations, as made precise below.

\begin{lemma}
    \label{claim:induced-CE}
    Let $\vmu = \frac{1}{T} \sum_{t=1}^T \vx^{(t)} \otimes \vy^{(t)}$ be an $\epsilon'$-CE of $(\mat{R}', \mat{C}')$ with $k \geq 12T$ and $\epsilon' \leq \frac{1}{2}$. Suppose further that $\setT \subseteq [T]$ is defined as $\setT \defeq \{ t \in [T] : \sumx^{(t)} \geq 1 - 6T/k \land \sumy^{(t)} \geq 1 - 6 T/k \}$. If $\setT \neq \emptyset$, then $\widehat{\vmu} = \frac{1}{|\setT|} \sum_{t \in \setT} \hvx^{(t)} \otimes \hvy^{(t)}$ is an $(\epsilon' T + 30 T^2/k)$-CE of $(\mat{R}, \mat{C})$, where $\hvx^{(t)}_{} = \nicefrac{\vx^{(t)}_{i \leq n}}{\sumx^{(t)}}$ and $\hvy_{}^{(t)} = \nicefrac{\vy^{(t)}_{j \leq n}}{\sumy^{(t)}}$.
\end{lemma}

\begin{proof}
    By~\Cref{claim:off-diagonal}, we know that for $t \in [T] \setminus \setT$ it holds that $\sumx^{(t)} \leq 6 T/k$ and $\sumy^{(t)} \leq 6 T/k$. We assume that $\setT \neq \emptyset$. We consider any deviation function of Player $x$ captured by $\widehat{\phi}_x : [n] \to [n]$. We extend $\widehat{\phi}_x$ to $[2n] \to [2n]$ as follows.
    \[
        \phi_x(i) = 
        \begin{cases}
            \widehat{\phi}_x(i) & i \leq n,\\
            i & i \geq n+1.
        \end{cases}
    \]
    It is clear that $\phi_x$ indeed belongs to the set $\Phi$ of considered deviations (recall the definition in~\Cref{sec:prels}). Since $\vmu$ is assumed to be an $\epsilon$-CE of $(\mat{R}', \mat{C}')$, we have
    \begin{equation}
        \label{eq:def-CE-1}
        \frac{1}{T} \sum_{t=1}^T \langle \vx^{(t)}, \mat{R}' \vy^{(t)} \rangle \geq \frac{1}{T} \sum_{t=1}^T \langle \phi_x(\vx^{(t)}), \mat{R}' \vy^{(t)} \rangle - \epsilon',
    \end{equation}
where $\phi_x(\vx^{(t)}) = ( \widehat{\phi}_x(\vx_{i \leq n}^{(t)}), \vx^{(t)}_{i\geq n+1})$ (by definition of $\widehat{\phi}_x$). As a result,
\begin{align}
    \frac{1}{T} \sum_{t = 1}^T \langle \vx^{(t)} - \phi_x(\vx^{(t)}), \mat{R}' \vy^{(t)} \rangle &= \frac{1}{T} \sum_{t =1}^T \langle \vx_{i \leq n}^{(t)} - \widehat{\phi}_x(\vx_{i \leq n}^{(t)}), \mat{R} \vy^{(t)}_{j \leq n} - k \sum_{j=n+1}^{2n} \vy_{j}^{(t)} \vec{1}_n \rangle \notag \\
    &= \frac{1}{T} \sum_{t = 1}^T \langle \vx_{i \leq n}^{(t)} - \widehat{\phi}_x(\vx_{i \leq n}^{(t)}), \mat{R} \vy^{(t)}_{j \leq n} \rangle \label{align:phihat}
\end{align}
since $\vec{1}_n^\top \vx_{i \leq n}^{(t)} - \vec{1}_n^\top \widehat{\phi}_x(\vx_{i \leq n}^{(t)}) = 0$. We now treat products in $[T] \setminus \setT$. We have
\begin{align}
    \frac{1}{T} \sum_{t \notin \setT} \langle \vx_{i \leq n}^{(t)} - \widehat{\phi}_x(\vx_{i \leq n}^{(t)}), \mat{R} \vy^{(t)}_{i \leq n} \rangle &\leq \frac{1}{T} \sum_{t \notin \setT} \| \vx^{(t)}_{i \leq n} - \widehat{\phi}_x(\vx^{(t)}_{i \leq n}) \|_1 \|\vy^{(t)}_{i \leq n} \|_1
    \label{align:Rbound} \\
    &\leq \frac{2}{T} \sum_{t \notin \setT} \|\vx^{(t)}_{i \leq n} \|_1 \|\vy^{(t)}_{i \leq n} \|_1 \leq \frac{72 T^2}{k^2},\label{align:col-tri}
\end{align}
where~\eqref{align:Rbound} uses the assumption that the entries of $\mat{R}$ are in $[-1, 1]$, and~\eqref{align:col-tri} uses the triangle inequality together with $\| \widehat{\phi}_x(\vx^{(t)}_{i \leq n}) \|_1 = \| \vx_{i \leq n}^{(t)} \|_1$, as well as the fact that $\sumx^{(t)}, \sumy^{(t)} \leq 6T/k$ for all $t \notin \setT$. Further, for products in $\setT$,
\begin{equation}
    \label{eq:setT}
    \frac{1}{T} \sum_{t \in \setT} \langle \vx_{i \leq n}^{(t)} - \widehat{\phi}_x(\vx_{i \leq n}^{(t)}), \mat{R} \vy^{(t)}_{j \leq n} \rangle \leq \frac{1}{T} \sum_{t \in \setT} \langle \hvx^{(t)} - \widehat{\phi}_x(\hvx^{(t)}), \mat{R} \hvy^{(t)} \rangle + \frac{24 T}{k},
\end{equation}
where we used that
\[
    \| \vx^{(t)}_{i \leq n} - \hvx^{(t)} \|_1 \leq 1 - \sumx^{(t)} \leq \frac{6 T}{k},
\]
and similarly for $\| \vy^{(t)}_{j \leq n} - \hvy^{(t)} \|_1$. Combining~\eqref{eq:def-CE-1}, \eqref{align:phihat}, \eqref{align:col-tri}, and \eqref{eq:setT}, we get
\[
    \frac{1}{|\setT|} \sum_{t \in \setT} \langle \hvx^{(t)} - \widehat{\phi}_x(\hvx^{(t)}), \mat{R} \hvy^{(t)} \rangle \geq - \epsilon' T - \frac{30 T^2}{k}.
\]
Similar reasoning bounds deviations by Player $y$, leading to the claim.
\end{proof}

Based on this claim, we next show that, in fact, $\setT = \emptyset$. This relies on the property of the $\soshard$ game $(\mat{R}, \mat{C})$, as made precise below.

\begin{lemma}
    \label{claim:Lsmall}
    Let $\vmu = \frac{1}{T} \sum_{t=1}^T \vx^{(t)} \otimes \vy^{(t)}$ be an $\epsilon'$-CE of $(\mat{R}', \mat{C}')$ with $k \geq 60 T/\epsilon'$, $\epsilon' \leq \epsilon_0/(2T)$ and $\epsilon = 2 \epsilon'$ per~\Cref{def:soshard}. Then, it holds that $\sumx^{(t)}, \sumy^{(t)} \leq \epsilon'$ for any $t \in [T]$.
\end{lemma}

\begin{proof}
    It suffices to show that $\setT = \emptyset$. For the sake of contradiction, let us suppose otherwise. Then, by~\Cref{claim:induced-CE}, we know that $\widehat{\vmu} = \frac{1}{|\setT|} \sum_{t \in \setT} \hvx^{(t)} \otimes \hvy^{(t)}$ is an $(\epsilon' T + 30 T^2/k)$-CE of $(\mat{R}, \mat{C})$. Thus, by our choice of parameters, it follows that $\hvmu$ is an $\epsilon_0$-CE of $(\mat{R}, \mat{C})$. By the property of $(\mat{R}, \mat{C})$ as an $\soshard$ game, it follows that $\hvmu$ attains welfare at most $2 \delta - 2 \epsilon$. Thus, under $\hvmu$, one of the two players obtains a utility of at most $\delta - \epsilon$; without loss of generality, we assume that $\frac{1}{|\setT|} \sum_{t \in \setT} \langle \hvx^{(t)}, \mat{R} \hvy^{(t)} \rangle \leq \delta - \epsilon$.
    
    We next consider the deviation $\phi_x(\vx^{(t)}) = (\vec{0}, \vxdef) \sumx^{(t)} + (\vec{0}, \vx^{(t)}_{i \geq n+1})$. Then, for any $t \in [T]$,
\begin{align*}
    \langle \vx^{(t)} - \phi_x(\vx^{(t)}), \mat{R}' \vy^{(t)} \rangle &= \langle \vx^{(t)}_{i \leq n}, \mat{R} \vy^{(t)}_{j \leq n} - k \sum_{j=n+1}^{2n} \vy_j^{(t)} \vec{1}_n \rangle - \delta \sumx^{(t)} \sumy^{(t)} - \sumx^{(t)} \langle \vxdef, \mat{R}^S \vy_{j \geq n+1}^{(t)} \rangle \\
    &\leq \langle \vx^{(t)}_{i \leq n}, \mat{R} \vy^{(t)}_{j \leq n} \rangle -  \delta \sumx^{(t)} \sumy^{(t)},
\end{align*}
where we used the property of $\vxdef$ concerning $(\mat{R}^S, \mat{C}^S)$ in the inequality above. In particular, for $t \in \setT$,
\begin{align*}
    \frac{1}{T} \sum_{t \in \setT} \langle \vx^{(t)} - \phi_x(\vx^{(t)}), \mat{R}' \vy^{(t)} \rangle &\leq \frac{1}{T} \sum_{t \in \setT} \langle \hvx^{(t)} \mat{R} \hvy^{(t)} \rangle - \delta \frac{|\setT|}{T} + \frac{36 T}{k}
\end{align*}
since $\sumx^{(t)}, \sumy^{(t)} \geq 1 - 6 T/k$ for any $t \in \setT$. But given that $\frac{1}{|\setT|} \sum_{t \in \setT} \langle \hvx^{(t)}, \mat{R} \hvy^{(t)} \rangle \leq \delta - \epsilon$, we have
\[
    \frac{1}{T} \sum_{t \in \setT} \langle \vx^{(t)} - \phi_x(\vx^{(t)}), \mat{R}' \vy^{(t)} \rangle \leq - \epsilon + \frac{36 T}{k}.    
\]
Further,
\[
\frac{1}{T} \sum_{t \notin \setT} \langle \vx^{(t)} - \phi_x(\vx^{(t)}), \mat{R}' \vy^{(t)} \rangle \leq \frac{6 T}{k}.
\]
As a result, by our choice of parameters, we have
\[
    \frac{1}{T} \sum_{t = 1}^T \langle \vx^{(t)} - \phi_x(\vx^{(t)}), \mat{R}' \vy^{(t)} \rangle < - \epsilon',
\]
which is a contradiction since $\vmu$ is assumed to be an $\epsilon'$-CE.
\end{proof}

\begin{lemma}[Soundness]
    \label{lemma:stitch-soundness}
    Under the assumptions of~\Cref{claim:Lsmall}, let $\vmu = \frac{1}{T} \sum_{t=1}^T \vx^{(t)} \otimes \vy^{(t)}$ be an $\epsilon'$-CE of $(\mat{R}', \mat{C}')$. If $\hvx^{(t)} = \nicefrac{\vx^{(t)}_{i \geq n+1}}{(1 - \sumx^{(t)})}$ and $\hvy^{(t)} = \nicefrac{\vy^{(t)}_{j \geq n+1}}{(1 - \sumy^{(t)})}$ for all $t \in [T]$, then $\widehat{\vmu} = \frac{1}{T} \sum_{t=1}^T \hvx^{(t)} \otimes \hvy^{(t)}$ is a $5 \epsilon'$-CE of $(\mat{R}^S, \mat{C}^S)$. Further,
    \[
        \left\| \frac{1}{T} \sum_{t=1}^T \vx^{(t)} - \frac{1}{T} \sum_{t=1}^T \hvx^{(t)} \right\|_1 \leq 2 \epsilon'.
    \]
\end{lemma}

\begin{proof}
    Consider any deviation $\Phi_{\text{int}} \ni \widehat{\phi}_x : \{n+1, \dots, 2n\} \to \{n+1, \dots, 2n \}$. We extend $\widehat{\phi}_x$ to $[2n] \to [2n]$ as follows.
    \[
        \phi_x(i) = 
        \begin{cases}
            i & i \leq n, \\
            \widehat{\phi}_x(i) & i \geq n+1.
        \end{cases}
    \]
    It is clear that $\phi_x$ is also an internal deviation. Since $\vmu$ is assumed to be an $\epsilon'$-CE of $(\mat{R}', \mat{C}')$, we have
\[
    \frac{1}{T} \sum_{t=1}^T \langle \vx^{(t)}, \mat{R}' \vy^{(t)} \rangle \geq \frac{1}{T} \sum_{t=1}^T \langle \phi_x(\vx^{(t)}), \mat{R}' \vy^{(t)} \rangle - \epsilon',
\]
where $\phi_x(\vx^{(t)}) = ( \vx_{i \leq n}^{(t)}, \widehat{\phi}_x(\vx^{(t)}_{i\geq n+1}))$. Thus, 
    \begin{align*}
    \frac{1}{T} \sum_{t=1}^T \langle \vx^{(t)} - \phi_x(\vx^{(t)}), \mat{R}' \vy^{(t)} \rangle &= \frac{1}{T} \sum_{t=1}^T \langle \vx^{(t)}_{i \geq n+1} - \widehat{\phi}_x(\vx^{(t)}_{i \geq n+1}), \delta \sumy^{(t)} \vec{1}_n + \mat{R}^S \vy_{j \geq n+1}^{(t)} \rangle \\
    &= \frac{1}{T} \sum_{t=1}^T \langle \vx^{(t)}_{i \geq n+1} - \widehat{\phi}_x(\vx^{(t)}_{i \geq n+1}), \mat{R}^S \vy_{j \geq n+1}^{(t)} \rangle \\
    &\leq \frac{1}{T} \sum_{t=1}^T \langle \hvx^{(t)} - \widehat{\phi}_x(\hvx^{(t)}), \mat{R}^S \hvy^{(t)} \rangle + 4 \epsilon',
\end{align*}
where we used the fact that $\sumx^{(t)}, \sumy^{(t)} \leq \epsilon'$ for all $t \in [T]$ (by~\Cref{claim:Lsmall}). When $\widehat{\phi}_x \notin \Phi_{\text{int}}$, one can extend $\widehat{\phi}_x$ to $[2n] \to [2n]$ so that
\[
    \phi_x(i) = \begin{cases}
        i' & i \leq n, \\
        \widehat{\phi}_x(i) & i \geq n+1,
    \end{cases}
\]
where $i' \geq n+1$ is selected such that $\phi_x \in \Phi$. Similar reasoning yields that
\[
    \frac{1}{T} \sum_{t=1}^T \langle \vx^{(t)} - \phi_x(\vx^{(t)}), \mat{R}' \vy^{(t)} \rangle \leq \frac{1}{T} \sum_{t=1}^T \langle \hvx^{(t)} - \widehat{\phi}_x(\hvx^{(t)}), \mat{R}^S \hvy^{(t)} \rangle + 6 \epsilon'.
\]
The proof with respect to Player $y$ is analogous, implying that $\widehat{\vmu}$ is a $7 \epsilon'$-CE for $(\mat{R}^S, \mat{C}^S)$. This completes the proof.
\end{proof}

\paragraph{Putting everything together} Finally, \Cref{theorem:game-collection-high} follows by combining~\Cref{theorem:soshard-high,theorem:enumhard-high} with~\Cref{lemma:stitch-soundness,lemma:stitch-completeness}; \Cref{theorem:game-collection-low} follows by combining~\Cref{theorem:soshard-low,theorem:enumhard-low} with~\Cref{lemma:stitch-soundness,lemma:stitch-completeness}.
\section{Conclusions and future research}
\label{sec:concl}

Our main contribution here was to provide evidence that existing regret minimization algorithms are close to optimal in terms of the number of iterations needed to approximate an equilibrium. There are many interesting avenues for future research. Besides the obvious question of further sharpening the lower bounds we obtain, an important question is to characterize the class of learning algorithms captured by low-degree SoS relaxations; is it the case that common no-regret dynamics can be understood through the algorithmic framework followed in our work? Furthermore, while the lower bounds we obtain have the advantage of being unconditional, they only apply to a restricted class of algorithms. It is natural to conjecture that computing sparse CE is \PPAD-hard (for a certain regime of sparsity), but that remains an open question.

\section*{Acknowledgments}

This material is based on work supported by the Vannevar Bush Faculty Fellowship ONR N00014-23-1-2876, National Science Foundation grants RI-2312342 and RI-1901403, ARO award W911NF2210266, and NIH award A240108S001. Alkis Kalavasis was supported by the Institute for Foundations of Data Science at Yale.

\printbibliography

\end{document}